%% file: main.tex
\documentclass{lmcs}
\pdfoutput=1

\usepackage{lastpage}
\lmcsdoi{16}{1}{23}
\lmcsheading{}{\pageref{LastPage}}{}{}%
{Jan.~31,~2019}{Feb.~20,~2020}{}

\usepackage{booktabs}
\input{packages}

\input{macros}

\renewenvironment{proof}{\lmcsproof}{\hfill\qed}

\begin{document}

\title[Decreasing Diagrams for Confluence and Commutation]{%
  Decreasing Diagrams\\ for Confluence and Commutation%
}

\author[Endrullis]{J\"{o}rg Endrullis\rsuper{a}}
\address{%
  \lsuper{a}Vrije Universiteit Amsterdam,
  Department of Computer Science,
  Amsterdam, the Netherlands
}
\address{%
  \lsuper{b}Centrum Wiskunde \& Informatica (CWI),
  Amsterdam, the Netherlands
}

\author[Klop]{Jan Willem Klop\rsuper{{a,b}}}
\email{j.endrullis@vu.nl}
\email{j.w.klop@vu.nl}
\email{r.overbeek@vu.nl}

\author[Overbeek]{Roy Overbeek\rsuper{{a,b}}}

%

\maketitle

\begin{abstract}
  \input{abstract}
\end{abstract}

\section{Introduction}\label{sec:intro}
\input{intro}

\section{Preliminaries}
\input{prelims}

\section{First-order Definability of Rewriting Properties}\label{sec:first:order}
\input{first_order}

\section{Decreasing Diagrams for Confluence with Two Labels}\label{sec:two}
\input{two}

\section{Decreasing Diagrams for Commutation}\label{sec:commutation}
\input{commutation}


\section{Conclusion}\label{sec:conclusion}
\input{conclusion}

\bibliographystyle{alpha}
\bibliography{main}

\end{document}

%% file: packages.tex
\usepackage{etex}
\let\lmcsproof=\proof
\usepackage{proof}
\usepackage{amsmath,amssymb,amstext}
\usepackage{wrapfig}
\usepackage{framed}
\usepackage{url}
\usepackage{tikz}
\usepackage{enumitem}
\usepackage{enumitem}
\usepackage{proof}
\usepackage{wrapfig}

\usepackage{tikz}
\usetikzlibrary{arrows,automata,backgrounds,calc,chains,fadings,folding,decorations.fractals,fit,patterns,positioning,mindmap,shadows,shapes.geometric,shapes.symbols,through,trees,decorations.markings,decorations.text}
\colorlet{dblue}{blue!40!black}

\input{tikz_settings.xet}
\usetikzlibrary{decorations.pathmorphing}
\usepackage{url}
\usepackage{subcaption}
\usepackage{float}

\colorlet{lightblue}{blue!50!white}

%% file: macros.tex
\newcommand{\extend}[1]{}

\theoremstyle{plain}
\newtheorem{theorem}[thm]{Theorem}
\newtheorem{corollary}[thm]{Corollary}
\newtheorem{lemma}[thm]{Lemma}
\newtheorem{proposition}[thm]{Proposition}

\newtheorem{open}[thm]{Open Problem}

\theoremstyle{definition}
\newtheorem{definition}[thm]{Definition}
\newtheorem{remark}[thm]{Remark}
\newtheorem{example}[thm]{Example}
\newtheorem{notation}[thm]{Notation}

\newcommand{\step}{\rightarrow}
\newcommand{\lstep}{\leftarrow}
\newcommand{\steps}{\twoheadrightarrow}
\newcommand{\lsteps}{\twoheadleftarrow}
\newcommand{\lrstep}{\leftrightarrow}
\newcommand{\lrsteps}{\leftrightarrow^*}

\newcommand{\nat}{\mathbb{N}}
\newcommand{\sdist}{d}
\newcommand{\dist}[1]{\sdist(#1)}
\newcommand{\distm}[2]{\sdist(#1,#2)}
\renewcommand{\fam}[1]{\{#1\}}
\newcommand{\lab}[1]{#1_{\{0,1\}}}

\colorlet{corange}{orange!70!red}
\colorlet{cred}{red}
\colorlet{cdarkred}{red!70!black}
\colorlet{clightred}{red!70}
\colorlet{cpurple}{blue!50}
\definecolor{cblue}{rgb}{0,0.4,0.7}
\definecolor{clightblue}{rgb}{0,0.6,1.0}
\colorlet{cgreen}{green!60!black}
\colorlet{cdarkgreen}{green!60!black}

\tikzset{ver/.style={blue!80!black,nodes={black}}}
\tikzset{horsn/.style={decorate,decoration={snake,amplitude=.4mm,segment length=2mm,post length=1mm}}}
\tikzset{hor/.style={cgreen,nodes={black},horsn}} 
\tikzset{exi/.style={densely dotted}}

\makeatletter
\tikzoption{xabove of}{\mytikz@of{#1}{90}}%
\tikzoption{xbelow of}{\mytikz@of{#1}{-90}}%
\tikzoption{xleft of}{\mytikz@of{#1}{180}}%
\tikzoption{xright of}{\mytikz@of{#1}{0}}%
\tikzoption{xabove left of}{\mytikz@of{#1}{135}}%
\tikzoption{xbelow left of}{\mytikz@of{#1}{-135}}%
\tikzoption{xabove right of}{\mytikz@of{#1}{45}}%
\tikzoption{xbelow right of}{\mytikz@of{#1}{-45}}%
\def\mytikz@of#1#2{%
  \def\tikz@anchor{180+#2}%
  \let\tikz@do@auto@anchor=\relax%
  \tikz@addtransform{%
    \expandafter\tikz@extract@node@dist\tikz@node@distance and\pgf@stop%
    \pgftransformshift{\pgfpointpolar{#2}{\tikz@extracted@node@distance - 5mm and 5mm}}}%
  \def\tikz@node@at{\pgfpointanchor{#1}{#2}}}
\def\tikz@extract@node@dist#1and#2\pgf@stop{%
  \def\tikz@extracted@node@distance{#1}}
\makeatother

\newcommand{\ver}{\to}
\newcommand{\hor}{\mathrel{{\color{cgreen}\leadsto}}}
\newcommand{\veri}{\lstep}
\newcommand{\hori}{\mathrel{{\color{cgreen}\reflectbox{$\leadsto$}}}}

\newcommand{\fun}{\mathsf}

\newcommand{\real}{\mathbb{R}}
\newcommand{\inter}[1]{[{#1}]}
\newcommand{\set}[1]{\{\, #1 \,\}}

\newcommand{\card}[1]{\#\, #1}

\newcommand{\nf}[1]{\mathit{nf}(#1)}

\renewcommand{\prop}{\textit}
\newcommand{\SN}{\ensuremath{\prop{SN}}}
\newcommand{\WN}{\ensuremath{\prop{WN}}}
\renewcommand{\CR}{\ensuremath{\prop{CR}}}
\newcommand{\WCR}{\ensuremath{\prop{WCR}}}
\newcommand{\UN}{\ensuremath{\prop{UN}}}
\newcommand{\UNrew}{\ensuremath{\prop{UN}^\to}}
\newcommand{\NF}{\ensuremath{\prop{NFP}}}
\newcommand{\SC}{\ensuremath{\prop{SC}}}
\newcommand{\AC}{\ensuremath{\prop{AC}}}
\newcommand{\IND}{\ensuremath{\prop{IND}}}
\newcommand{\INC}{\ensuremath{\prop{INC}}}

\newcommand{\DCR}{\ensuremath{\prop{DCR}}}
\newcommand{\DC}{\ensuremath{\prop{DC}}}
\newcommand{\CP}{\ensuremath{\prop{CP}}}
\newcommand{\CNT}{\ensuremath{\prop{CNT}}}



%% file: abstract.tex
Like termination, confluence is a central property of rewrite systems. 
Unlike for termination, however, there exists no known complexity hierarchy for confluence.
In this paper we investigate whether the decreasing diagrams technique can be used to obtain such a hierarchy. 
The decreasing diagrams technique is one of the strongest and most versatile methods for proving confluence of abstract rewrite systems.
It is complete for countable systems, and it has many well-known confluence criteria as corollaries. 

So what makes decreasing diagrams so powerful? 
In contrast to other confluence techniques, 
decreasing diagrams employ a labelling of the steps with labels from a well-founded order 
in order to conclude confluence of the underlying unlabelled relation.
Hence it is natural to ask how the size of the label set influences the strength of the technique.
In particular, what class of abstract rewrite systems can be proven 
confluent using decreasing diagrams restricted to $1$~label, $2$ labels, $3$ labels, and so on?
Surprisingly, we find that two labels suffice for proving confluence for every abstract rewrite system having the cofinality property, 
thus in particular for every confluent, countable system. 

Secondly, we show that this result stands in sharp contrast to the situation for commutation of rewrite relations, where the hierarchy does not collapse.

Thirdly, investigating the possibility of a confluence hierarchy,
we determine the first-order (non-)definability of the notion of confluence
and related properties, using techniques from finite model theory.
We find that in particular Hanf's theorem 
is fruitful for elegant proofs of undefinability of properties of abstract rewrite systems.


%% file: intro.tex
A binary relation~$\to$ is called \emph{confluent}
if two coinitial reductions (i.e., reductions having the same starting term) can always be extended towards a common reduct, that is:
\begin{align}
  \forall abc.\; \big(b \lsteps a \steps c \Rightarrow \exists d.\; b \steps d \lsteps c\big)\;.
  \label{eq:confluence}
\end{align}
The confluence property is illustrated in Figure~\ref{fig:confluence},
in which solid and dotted lines stand for universal and existential quantification,
respectively.
The relation $\to$ is called \emph{terminating} if there are no infinite sequences $a_0 \to a_1 \to a_2 \to \ldots$. 

\begin{figure}[t!]
  \centering
  \begin{minipage}{.49\textwidth}
    \centering
    \begin{tikzpicture}[default,node distance=15mm]
      \node (a) {$a$};
      \node (b) [below left of=a] {$b$};
      \node (c) [below right of=a] {$c$};
      \node (d) [below right of=b] {$d$};
      
      \begin{scope}[->>]
        \draw (a) -- (b);
        \draw (a) -- (c);
        \draw [exi] (b) -- (d);
        \draw [exi] (c) -- (d);
      \end{scope}
    \end{tikzpicture}\vspace{-1ex}
    \caption{Confluence.}
    \label{fig:confluence}
  \end{minipage}  
  \begin{minipage}{.49\textwidth}
    \centering
    \begin{tikzpicture}[default,node distance=15mm]
      \node (a) {$a$};
      \node (b) [below left of=a] {$b$};
      \node (c) [below right of=a] {$c$};
      \node (d) [below right of=b] {$d$};
      
      \begin{scope}[->>]
        \draw [hor] (a) -- (b);
        \draw (a) -- (c);
        \draw [exi] (b) -- (d);
        \draw [exi,hor] (c) -- (d);
      \end{scope}
    \end{tikzpicture}\vspace{-1ex}
    \caption{Commutation.}
    \label{fig:def:commutation}
  \end{minipage}
\end{figure}

Termination and confluence are central properties of rewrite systems.
For both properties there exist numerous proof techniques,
and there are annual competitions for comparing the performance of automated provers \cite{beyer2019tools}. 
It is therefore a natural question how to measure and classify the complexity of termination and confluence problems.
While there is a well-known hierarchy for termination~\cite{zant:2001},
no such classification is known for confluence.\footnote{%
  Ketema and Simonsen~\cite{kete:simo:2013} 
  consider peaks $t_1 \lsteps s \steps t_2$ and measure the length of joining reductions $t_1 \steps \cdot \lsteps t_2$  
  as a function of the size of $s$ and the length of the reductions in the peak.
  The nature of this function can serve as a complexity measure for a confluence problem.
} 

The termination hierarchy~\cite{zant:2001} is based on 
the characterisation of termination in terms of well-founded monotone algebras.
This entails an interpretation of the symbols of the signature as functions over the algebra.
Then the class of the functions (or other properties of the algebra) used to establish termination can serve as a measure for the complexity of the termination problem.
For instance, if polynomial functions over the natural numbers suffice to establish termination, then the rewrite system is said to be polynomial terminating.\extend{\footnote{%
  There are systems with linear-bounded reduction length for which
  termination cannot be proven using a polynomial interpretation,
  for instance the string rewrite system $\{\, \fun{aa} \to \fun{aba} \,\}$;
  see further~\cite{zant:2001}.
}}

In order to address the question of a hierarchy and complexity measure for the confluence property, 
our point of departure is the decreasing diagrams technique~\cite{oost:1994b}.
Decreasing diagrams are for confluence what well-founded interpretations are for termination.
The decreasing diagrams technique is complete for 
systems having the cofinality property \cite[p.\ 766]{tere:2003}.
Thus, in particular for every confluent, countable abstract rewrite system,
the confluence property can be proven using the decreasing diagrams technique. 
The power of decreasing diagrams is moreover witnessed by the fact that
many well-known confluence criteria are direct consequences of decreasing diagrams~\cite{oost:1994b},
including the lemma of Hindley--Rosen~\cite{hind:1964,rose:1973},
Rosen's request lemma~\cite{rose:1973},
Newman's lemma~\cite{newm:1942}, and
Huet's strong confluence lemma~\cite{huet:1980}.

\emph{What makes the decreasing diagrams technique so powerful?}
The freedom to label the steps distinguishes decreasing diagrams 
from all other confluence criteria,
with the exception of the weak diamond property~\cite{brui:1978,endr:klop:2013} by De Bruijn which has equal strength. 
This suggests that the power of these techniques arises from the labelling.
This naturally leads to the following questions:
\begin{enumerate}
  \item How does the size of the label set influence the strength of decreasing diagrams?
  \item What class of abstract rewrite systems can be proven confluent 
        using decreasing diagrams with $1$ label, $2$ labels, $3$ labels and so on?
  \item Can the size of the label set serve as a complexity measure for a confluence problem?
\end{enumerate}
Let $\DCR$ denote the class of abstract rewrite systems (ARSs) whose confluence can be proven using decreasing diagrams.
For an ordinal $\alpha$,
we write $\DCR_\alpha$ for the class of ARSs
whose confluence can be proven using decreasing diagrams with label set $\alpha$ (see Definition~\ref{def:dcr:alpha}).

For every ARS $\aars$, we have
\begin{align}
  \DCR(\aars) \;\implies\; \text{$\DCR_\alpha(\aars)$ for some ordinal $\alpha$} 
\end{align}
The reason is that any partial well-founded order can be transformed into a total well-founded order (thus an ordinal).
This transformation does not require the Axiom of Choice, see~\cite{endr:klop:2013}.

Clearly, we have $\DCR_\alpha \subseteq \DCR_\beta$ whenever $\alpha < \beta$. So
\begin{align}
  \DCR_1 \subseteq \DCR_2 \subseteq \DCR_3 \subseteq \ldots \subseteq \DCR_\omega \subseteq \ldots 
  \label{eq:hierarchy:dcr}
\end{align}%
\emph{But which of these inclusions are strict?}
From the completeness proof in~\cite{oost:1994} it follows that 
all abstract rewrite systems having the \emph{cofinality property}, including all countable systems, belong to $\DCR_\omega$.
In other words, for confluence of countable systems it suffices to 
label steps with natural numbers.

As we are investigating a confluence hierarchy,
the question of first-order definability of confluence arises naturally.
Namely, if confluence were definable by a set of first-order formulas, 
then we could obtain a confluence hierarchy by imposing syntactic restrictions on this set of formulas.
To this end, we investigate first-order definability of confluence and related properties in Section~\ref{sec:first:order}.

\subsection*{Contribution and outline}

We start by investigating the definability of various first-order properties of rewrite systems in Section~\ref{sec:first:order}. We show that most of the considered properties are not first-order definable (assuming an equality relation and the \emph{one-step} rewrite relation), in part by applying methods from the field of finite model theory.
 
Our main result is that all systems with the cofinality property are in the class $\DCR_2$, see Section~\ref{sec:two}.
In particular, for proving confluence of countable abstract rewrite systems
it always suffices to label steps with $0$ or $1$ using the order $0 < 1$.
So for countable systems, the hierarchy~\eqref{eq:hierarchy:dcr} collapses at level $\DCR_2$.
This is somewhat surprising, 
as one might expect that the method of decreasing diagrams draws its strength from a rich labelling of the steps.

Interestingly, there is a stark contrast with commutation. 
For commutation the hierarchy does not collapse, see  Section~\ref{sec:commutation}.
We prove that, for commutation of countable systems, all inclusions are strict up to level $\DC_\omega$.

Our findings also provide new ways to approach the long-standing open problem 
of completeness of decreasing diagrams for uncountable systems, see~Section~\ref{sec:conclusion}.


%% file: prelims.tex
We repeat some of the main definitions, for the sake of self-containedness, and to fix notations.
Let $A$ be a set.
For a relation ${\to} \subseteq A\times A$ we write 
\begin{enumerate}
  \item $\to^+$ for its transitive closure,
  \item $\to^*$ or $\steps$ for its reflexive transitive closure,
  \item $\lrstep$ for $\lstep \cup \step$; so $\lrsteps$ stands for convertibility, and
  \item $\equiv$ for the empty step, that is, ${\equiv} = \{(a,a) \mid a \in A\}$, and we define ${\to^\equiv} = {\to \cup \equiv}$.
\end{enumerate}

\begin{definition}[Abstract Rewrite System]
  An \emph{abstract rewrite system (ARS)} $\aars = (A,\to)$ consists of a non-empty set~$A$ 
  together with a binary relation ${\to} \subseteq A \times A$.
  For $B \subseteq A$ we define $\aars|_{B}$, the \emph{restriction of $\aars$ to $B$},
  by $\aars|_{B} = (B,\;{\to} \cap (B\times B))$.
\end{definition}

\begin{definition}[Indexed ARS]
  An \emph{indexed ARS} $\aars = \indars$ consists of a non-empty set $A$ of \emph{objects},
  and a family $\fam{\to_\alpha}_{\alpha\in I}$ of relations ${\to_\alpha} \subseteq A \times A$
  indexed by some set~$I$.
\end{definition}

\begin{definition}[Local Confluence]
  An ARS $(A,\to)$ is \emph{locally (or weakly) confluent (\WCR)} if ${\lstep \cdot \step} \subseteq {\steps \cdot \lsteps}$.
\end{definition}

\begin{definition}[Confluence]
  An ARS $(A,\to)$ is \emph{confluent (\CR)} if ${\lsteps \cdot \steps} \subseteq {\steps \cdot \lsteps}$,
  that is, every pair of finite, coinitial rewrite sequences can be joined to a common reduct. 
\end{definition}

\begin{definition}[Strong Confluence]
  An ARS $(A,\to)$ is \emph{strongly confluent} if ${\lstep \cdot \step} \subseteq {\step^\equiv \cdot \lsteps}$.
\end{definition}
Strong confluence is due to Huet~\cite{huet:1980}.
Note that ${\lstep \cdot \step} \subseteq {\step^\equiv \cdot \lsteps}$ is equivalent to ${\lstep \cdot \step} \subseteq {\steps \cdot \lstep^\equiv}$
as is clear by writing this property as
\begin{align*}
  \forall a x y.\; \exists z.\; (a \step x \wedge a \step y) \implies (x \step^\equiv z \lsteps y)
\end{align*}
and swapping $x,y$.
Thus there is freedom of choice in which side of the converging reduction `splitting' occurs -- which implies confluence.
Hence the name strong confluence.

\begin{definition}[Commutation]
  Let $(A,{\ver},{\hor})$ be an indexed ARS.
  Then the relation $\ver$ \emph{commutes with~$\hor$} if 
  ${\veri^* \cdot \hor^*} \subseteq {\hor^* \cdot \veri^*}$; 
  see Figure~\ref{fig:def:commutation}. 
\end{definition}

\begin{definition}[Normal Form]
Let $(A,\to)$ be an ARS. An $a \in A$ is a \emph{normal form} if there exists no $b \in A$ such that $a \to b$.
\end{definition}

\begin{definition}[Unique Normal Forms]
  An ARS $(A,\to)$ has \emph{unique normal forms (\UN)} if for all normal forms $a,b \in A$
  it holds that $a \lrsteps b \implies a = b$. 
\end{definition}

\begin{definition}[Unique Normal Forms with Respect to Reduction]
  An ARS $(A,\to)$ has \emph{unique normal forms with respect to reduction (\UNrew)} if for all normal forms $a,b \in A$
  it holds that $a \lsteps \cdot \steps b \implies a = b$. 
\end{definition}

\begin{definition}[Normal Form Property]
  An ARS $(A,\to)$ has the \emph{normal form property (\NF)} if for all $a \in A$ and all normal forms $b \in A$
  it holds that $a \lrsteps b \implies a \steps b$. 
\end{definition}

\begin{definition}[Weak Normalisation]
  Let $(A, \to)$ be an ARS. An $a \in A$ is \emph{weakly normalising} if $a \steps b$ for some normal form $b \in A$. The relation $\to$ is \emph{weakly normalising (WN)} if every $a \in A$ is weakly normalising.
\end{definition}

\begin{definition}[Strong Normalisation]
  Let $(A, \to)$ be an ARS. An $a \in A$ is \emph{strongly normalising} if every reduction sequence starting from $a$ is finite. The relation $\to$ is \emph{strongly normalising (SN)} if every $a \in A$ is strongly normalising.
\end{definition}

\begin{definition}[Acyclicity]
  An ARS $(A,\to)$ is \emph{acyclic (\AC)} if for all $a,b \in A$ we have $a \step^+ b \implies a \ne b$. 
\end{definition}

\begin{definition}[Inductive]
  An ARS $(A,\to)$ is \emph{inductive (\IND)} if for every infinite rewrite sequence $a_0 \step a_1 \step a_2 \step \cdots$
  there exists an $a \in A$ such that $a_i \steps a$ for every $i \in \nat$.
\end{definition}

This property and also the following are due to Nederpelt~\cite{nede:1973}, developed in the context of the Automath project.
\begin{definition}[Increasing]
  An ARS $(A,\to)$ is \emph{increasing (\INC)} if there is a map $f : A \to \nat$
  such that $f(a) < f(b)$ whenever $a,b \in A$ and $a \to b$.
  So the `value' of an element increases in a reduction step.
\end{definition}

\noindent
Nederpelt~\cite{nede:1973} has shown that $\IND \;\&\; \INC \implies \SN$.

\begin{definition}[Countable]
  An ARS $(A,\to)$ is \emph{countable (\CNT)}
  if there exists a surjective function from the set of natural numbers $\nat$ to $A$.
\end{definition}  

\begin{definition}[Cofinal Reduction]
  Let $\aars = (A,\to)$ be an ARS.
  A set $B \subseteq A$ is \emph{cofinal} in~$\aars$ 
  if for every $a \in A$ we have $a \steps b$ for some $b \in B$.
  A finite or infinite reduction sequence
  $b_0 \to b_1 \to b_2 \to \cdots$
  is \emph{cofinal} in~$\aars$
  if the set $B = \{\,b_i \mid i \ge 0\,\}$ is cofinal in~$\aars$.
\end{definition}

\begin{definition}[Cofinality Property]
  An ARS $\aars = (A,\to)$ has the \emph{cofinality property (\CP)}
  if for every $a \in A$, there exists a reduction
  $a \equiv b_0 \to b_1 \to b_2 \to \cdots$ that is cofinal in $\aars|_{\{b \,\mid\, a \steps b\}}$. 
\end{definition}  

\begin{lemma}\label{lem:cofinal}
  Let $\aars = (A,\to)$ be a confluent ARS and $a \in A$.
  If a rewrite sequence is cofinal in $\aars|_{\{b \,\mid\, a \steps b\}}$,
  then it is also cofinal in $\aars|_{\{b \,\mid\, a \lrstep^* b\}}$. \qed 
\end{lemma}  

\begin{theorem}[Klop~\cite{klop:1980}]\label{thm:klop}
  Every confluent countable ARS has the cofinality property. \qed
\end{theorem}

The countability of ARSs will be an important concern later on.
Therefore we mention the well-known fact that 
there is a counterexample for the reverse implication.
There are two simple proofs.
The first uses the fact that $\aleph_1$ is a regular cardinal, thus having cofinality $\aleph_1$.
We include the second proof from~\cite{klop:1980} for completeness sake.

\begin{example}[Counterexample]
  Let $U$ be an uncountable set, and let 
  $$A = \set{ X \subseteq U \mid \text{$X$ is finite}}$$ 
  Take the ARS $\aars = (A,\to)$ where $\to$ is defined by
  \begin{align*}
    X \to X \cup \set{y}
  \end{align*}
  for every $X \in A$ and $y \in U \setminus X$.
  Then it is easy to show that $\aars$ is $\CR$, but not  $\CP$.
\end{example}

%% file: first_order.tex
In this section, we study the definability of properties of abstract rewrite systems (graphs)
in first-order logic with equality and a predicate for the \emph{one-step} rewrite relation.
In particular, we establish the definability and undefinability results shown in Figure~\ref{fig:overview}. 

\newcommand{\yes}{\textbf{yes}}
\begin{figure}[b!]
  \centering
  \def\arraystretch{1.2}
  \begin{tabular}{@{}lll@{}}
    \toprule
    property $P$ & definability of $P$ & definability of $\neg P$ \\
    \midrule
    confluence  (\CR) & no (Theorem~\ref{thm:fo:cr:hanf}) & no (Theorem~\ref{thm:fo:cr:hanf}) \\
    local confluence  (\WCR) & no (Theorem~\ref{thm:fo:cr:hanf}) & no (Theorem~\ref{thm:fo:cr:hanf}) \\
    $\DCR$ and $\DCR_\alpha$ for $\alpha \ge 2$ & no (Theorem~\ref{thm:fo:dcr:hanf}) & no (Theorem~\ref{thm:fo:dcr:hanf}) \\
    strong confluence  (\SC) & no (Theorem~\ref{thm:fo:sc:hanf}) & no (Theorem~\ref{thm:fo:sc:hanf}) \\
    strong normalisation (\SN) & no (Theorem~\ref{thm:fo:sn:hanf}) & no (Theorem~\ref{thm:fo:sn:hanf}) \\
    weak normalisation (\WN) & no (Theorem~\ref{thm:fo:sn:hanf}) & no (Theorem~\ref{thm:fo:sn:hanf}) \\
    unique normal forms  (\UN) & \textbf{yes} (Theorem~\ref{thm:fo:un}) & no (Theorem~\ref{thm:fo:not:un}) \\
    unique normal forms  (\UNrew) & \textbf{yes} (Theorem~\ref{thm:fo:un}) & no (Theorem~\ref{thm:fo:not:un}) \\
    normal normal property  (\NF) & no (Theorem~\ref{thm:fo:cr:hanf}) & no (Theorem~\ref{thm:fo:cr:hanf}) \\
    acyclicity (\AC) & \textbf{yes} (Theorem~\ref{thm:fo:un}) & no (Theorem~\ref{thm:fo:neg:ac}) \\
    increasing (\INC) & no (Theorem~\ref{thm:fo:inc:hanf}) & no (Theorem~\ref{thm:fo:inc:hanf}) \\
    inductive (\IND) & no (Theorem~\ref{thm:fo:sn:hanf}) & no (Theorem~\ref{thm:fo:sn:hanf}) \\
        cofinality property (\CP) & no (Theorem~\ref{thm:fo:cr:hanf}) & no (Theorem~\ref{thm:fo:cr:hanf}) \\
    \bottomrule
  \end{tabular}
  \caption{First-order definability of properties of rewrite systems.
    Here \emph{yes} or \emph{no} refers to the definability as a general first-order property, that is, definable by a set of first-order formulas.}
  \label{fig:overview}
\end{figure}

\begin{notation}
  For an ARS $\aars$ and a set of first-order sentences $\Delta$, we write
  \begin{align*}
    \aars \models \Delta
  \end{align*}
  to denote that $\aars$ is a \emph{model} of $\Delta$, that is, $\aars$ satisfies all formulas in $\Delta$. 
  Likewise, for a property $P$ of abstract rewrite systems, we write $\aars \models P$ if $P$ holds in $\aars$.
\end{notation}

We define first-order properties in the setting of abstract rewriting with a single rewrite relation $\to$.
\begin{definition}
  A property $P$ of abstract rewrite systems is a \emph{first-order property} (\emph{fop})
  if there exists a sentence $\phi$ in first-order logic with equality and the predicate $\to$ (one-step rewriting)
  such that, for every ARS $\aars = (A,\to)$, $\aars \models P$ if and only if $ \aars \models \set{ \phi }$.
\end{definition}

\begin{definition}
  A property $P$ of abstract rewrite systems is a \emph{generalised first-order property} (\emph{gfop})
  if there exists a set $\Phi$ of sentences in first-order logic with equality and the predicate $\to$ (one-step rewriting)
  such that, for every ARS $\aars = (A,\to)$, $\aars \models P$ if and only if $ \aars \models \Phi$.
\end{definition}
\noindent
We say that a property $P$ is \emph{definable in first-order logic} if $P$ is a gfop.

At first glance this question may appear trivial since confluence is typically defined via the first-order formula~\eqref{eq:confluence}.
However, this formula involves the transitive closure $\steps$ of the one-step relation $\step$ which is itself not first-order definable,
as is well-known.
We show that confluence is not first-order definable over the one-step relation~$\step$.

\begin{remark}
  In~\cite{trei:1998} it is shown that the first-order theory of linear one-step rewriting is undecidable.
  In this paper it is mentioned as a conjecture that undecidable properties like confluence and weak termination 
  (see further~\cite{endr:geuv:simo:zant:2011,endr:geuv:zant:2009})
  cannot be expressed in the first-order logic of one-step rewriting.
\end{remark}

We will establish the negative results about $\neg \UN$, $\neg \UNrew$ and $\neg \AC$ using the compactness theorem~\cite{dale:1994}:
\begin{theorem}[Compactness]
  A set of first-order sentences $\Gamma$ has a model if and only if every finite subset of $\Gamma$ has a model.
\end{theorem}

For the other properties $P$, for which $P$ as well as $\neg P$ are undefinable,
we will employ Hanf's theorem, well-known in finite model theory.

\section*{Undefinability via Compactness}
  
In the following proofs, we write $\inter{c}$ for the interpretation of a constant $c$ in the model.
For convenience, we write $\to$ for the predicate symbol in formulas as well as for the actual one-step rewrite relation or $\aars$.
We use $\Rightarrow$ to denote implication in formulas.

\begin{theorem}\label{thm:fo:not:un}
  The properties $\neg \UN$, $\neg \UNrew$ and $\neg \NF$ are not gfops.
\end{theorem}

\begin{proof}
  Assume, for a contradiction, that there is a set $\Delta$ of first-order formulas over the predicate $\to$ such that 
  for every ARS $\aars = (A,\to)$ it holds that:
  \begin{align*}
    \text{$\aars$ is $\neg \NF$} \quad\iff\quad \aars \models \Delta 
  \end{align*}
  We describe the following non-confluent structure using formulas:
  \begin{center}
    \begin{tikzpicture}[default,->]
      \node (a) {$a$};
      \node (b0) [above right of=a] {$b_0$}; \draw (a) to (b0);
      \node (c0) [below right of=a] {$c_0$}; \draw (a) to (c0);

      \node (b1) [right of=b0] {$b_1$}; \draw (b0) to (b1);
      \node (b2) [right of=b1] {$b_2$}; \draw (b1) to (b2);
      \node (b3) [right of=b2] {$b_3$}; \draw (b2) to (b3);
      \node (b4) [right of=b3] {$\cdots$}; \draw (b3) to (b4);

      \node (c1) [right of=c0] {$c_1$}; \draw (c0) to (c1);
      \node (c2) [right of=c1] {$c_2$}; \draw (c1) to (c2);
      \node (c3) [right of=c2] {$c_3$}; \draw (c2) to (c3);
      \node (c4) [right of=c3] {$\cdots$}; \draw (c3) to (c4);
    \end{tikzpicture}
  \end{center} 
  We start by describing each single step by an atomic formula:
  \begin{align*}
    \Lambda = 
      \{\, a \to b_0,\; a \to c_0 \,\} 
      \cup \{\, b_i \to b_{i+1} \mid i \in \nat \,\}
      \cup \{\, c_j \to c_{j+1} \mid j \in \nat \,\}
  \end{align*}
  We need to ensure that the interpretation of distinct constants is distinct:
  \begin{align*}
    \Lambda_{\ne} = \{\, x \ne y \mid x,y \in N, x \ne y \,\} \quad \text{where} \quad N = \{\, a \,\} \cup \{\, b_i \mid i \in \nat \,\} \cup \{\, c_j \mid j \in \nat \,\}
  \end{align*}
  We need to ensure that $\inter{a}$ has at most two outgoing arrows:
  \begin{align*}
    \xi_a = \forall xyz.\; (a \to x \wedge a \to y \wedge a \to z) \Rightarrow (x = y \vee y = z \vee x = z) 
  \end{align*}
  Finally, the following formula requires all elements, except for $\inter{a}$, to be deterministic:
  \begin{align*}
    \xi_{\neg a} = \forall xyz.\; (x \ne a \wedge x \to y \wedge x \to z) \Rightarrow y = z 
  \end{align*}
  
  Now consider the following set of formulas:
  \begin{align*}
    \Gamma = \Delta \cup \Lambda \cup \Lambda_{\ne} \cup \{\, \xi_a, \xi_{\neg a} \,\}
  \end{align*}
  Let $\aars = (A,\to)$ be a model of $\Gamma$.
  Then $\aars$ is not $\NF$ since $\aars \models \Delta$.
  As a consequence, there exist $x,y \in A$ such that 
  $x$ is a normal form, $x \lrsteps y$ and $y \not\steps x$.
  Without loss of generality, we may assume that the conversion $x \lrsteps y$ is repetition-free, that is, no element occurs twice. 
  Hence for any peak $x' \lstep z \step y'$ in the conversion, $x' \neq y'$, and consequently $z = \inter{a}$ using $\aars \models \xi_{\neg a}$. 
  Since $x \lrsteps y$ is repetition-free, there can be at most one peak in the conversion. 
  Moreover, there must be at least one peak, for otherwise $x \steps \cdot \lsteps y$ and hence $y \steps x$ since $x$ is a normal form.
  Thus the conversion has exactly one peak and is of the form:
  \begin{align*}
    x \lsteps x' \lstep z \step y' \steps \cdot \lsteps y
  \end{align*}
  Then $x',y' \in \set{\inter{b_0},\inter{c_0}}$ since $\aars \models \xi_a$.
  However, the reduction graphs of $\inter{b_0}$ and $\inter{c_0}$ are both an infinite line (no branching)
  as a consequence of
  $\aars \models \Lambda \cup \Lambda_{\ne} \cup \xi_{\neg a} \cup \xi_a$.
  This implies that $\inter{b_0}$ and $\inter{c_0}$ have no normal forms,
  contradicting that $x$ is a normal form. 
  Hence $\Gamma$ has no model.
    
  On the other hand, any finite subset $\Gamma'$ of $\Gamma$ has a model.
  This can be seen as follows.
  There exists a $k \in \nat$ such that none of the constants
  $\{\, b_i \mid i \ge k \,\} \cup \{\, c_j \mid j \ge k \,\}$
  appears in $\Gamma'$.
  Then the following structure is a model of $\Gamma'$:
  \begin{center}
    \begin{tikzpicture}[default,->]
      \node (a) {$a$};
      \node (b0) [above right of=a] {$b_0$}; \draw (a) to (b0);
      \node (c0) [below right of=a] {$c_0$}; \draw (a) to (c0);

      \node (b1) [right of=b0] {$b_1$}; \draw (b0) to (b1);
      \node (b2) [right of=b1] {$b_2$}; \draw (b1) to (b2);
      \node (b3) [right of=b2] {$\cdots$}; \draw (b2) to (b3);
      \node (b4) [right of=b3] {$b_k$}; \draw (b3) to (b4);

      \node (c1) [right of=c0] {$c_1$}; \draw (c0) to (c1);
      \node (c2) [right of=c1] {$c_2$}; \draw (c1) to (c2);
      \node (c3) [right of=c2] {$\cdots$}; \draw (c2) to (c3);
      \node (c4) [right of=c3] {$c_k$}; \draw (c3) to (c4);

    \end{tikzpicture}
  \end{center} 
  This ARS does not have the property $\NF$ (even not $\UN$ or $\UNrew$).
  By the compactness theorem, this is a contradiction.
  Thus $\neg \NF$ is not first-order definable.
  
  The same proof also shows undefinability of $\neg \UN$ and $\neg \UNrew$.
  To this end, recall that $\NF \Rightarrow \UN \Rightarrow \UNrew$
  and thus $\neg \UNrew \Rightarrow \neg \UN \Rightarrow \neg \NF$.
\end{proof}

\begin{theorem}\label{thm:ac}
  The properties $\AC$ and $\neg \AC$ are not fops.
\end{theorem}

\begin{proof}
  This is a standard example in textbooks about finite model theory.
  See for instance~\cite{gradel2007finite,immerman2012descriptive,libkin2013elements}.
  These proofs use Ehrenfeucht-Fra\"{i}ss\'{e} games or a variant of Hanf's theorem.
\end{proof}

\begin{theorem}\label{thm:fo:un}
  The properties $\AC$, $\UN$ and $\UNrew$ are gfops, but not fops.
\end{theorem}

\begin{proof}
  We introduce the following abbreviations to denote formulas:
  \begin{align*}
    \nf{x} &\;=\; \neg \exists y.\; x \to y \\
    x \to^{0} y &\;=\; x = y \\
    x \to^{n+1} y &\;=\; \exists z.\; x \to z \wedge z \to^n y \\
    x \lrstep^{0} y &\;=\; x = y \\
    x \lrstep^{n+1} y &\;=\; \exists z.\; (x \to z \vee z \to x) \wedge z \lrstep^n y
  \end{align*}
  Define:
  \begin{align*}
    \Delta_{\UN} &= \set{ \forall a,b.\;\; \nf{a} \wedge \nf{b} \wedge a \lrstep^i b \;\Rightarrow\; a = b \mid  i \in \nat } \\
    \Delta_{\UNrew} &= \set{ \forall a,b,x.\;\; \nf{a} \wedge \nf{b} \wedge x \to^i a \wedge x \to^j b \;\Rightarrow\; a = b \mid  i,j \in \nat } \\
    \Delta_{\AC} &= \set{ \forall a,b.\;\; a \step^i b \;\Rightarrow\; a \neq b \mid  i > 0 }
  \end{align*}
  Then it is straightforward to verify that
  \begin{align*}
    \text{$\aars$ is \UN} \quad&\iff\quad \aars \models \Delta_{\UN} \\
    \text{$\aars$ is \UNrew} \quad&\iff\quad \aars \models \Delta_{\UNrew} \\
    \text{$\aars$ is \AC} \quad&\iff\quad \aars \models \Delta_{\AC}
  \end{align*}
  for every ARS $\aars = (A,\to)$.
  So the properties are definable by infinite sets of formulas.

  Note that $\UN$ and $\UNrew$ are not definable by single formulas
  since $\neg\UN$ and $\neg\UNrew$ are not.
  For $\AC$ this is established by Theorem~\ref{thm:ac}.
\end{proof}

If a property $P$ can be defined by a set of formulas, but not by a single formula,
then $\neg P$ cannot be defined by a set of formulas.

\begin{lemma}\label{lem:gfop:not:fop}
  If a property $P$ is a gfop, but not a fop, then $\neg P$ is not a gfop.
\end{lemma}

\begin{proof}
  Assume that there exists a set of formulas $\Delta_P$ such that
  \begin{align*}
    \text{$\aars$ has property $P$} \quad\iff\quad \aars \models \Delta_P
  \end{align*}
  for every ARS $\aars = (A,\to)$.
  For a contradiction, assume that there is a set $\Delta_{\neg P}$ of first-order formulas over the predicate $\to$ such that 
  \begin{align*}
    \text{$\aars$ has property $\neg P$} \quad\iff\quad \aars \models \Delta_{\neg P}
  \end{align*}
  for every ARS $\aars = (A,\to)$.
  Then $\Gamma = \Delta_{P} \cup \Delta_{\neg P}$ does not have a model.

  However, every finite subset $\Gamma' \subseteq \Gamma$ has a model. This can be seen as follows. 
  Define 
  $\Gamma'_{P} = \Gamma' \cap \Delta_{P}$ and 
  $\Gamma'_{\neg P} = \Gamma' \cap \Delta_{\neg P}$;
  then $\Gamma' = \Gamma'_{P} \cup \Gamma'_{\neg P}$.
  Assume that
  \begin{align}
    \text{$\aars$ has property $\neg P$} \quad\iff\quad \aars \models \Gamma'_{\neg P} \label{eq:finite:ac}
  \end{align}
  for every ARS $\aars = (A,\to)$.
  This yields a contradiction since $\Gamma'_{\neg P}$ is finite
  and consequently $\neg P$ could be characterised by a single formula (the conjunction of all formulas in $\Gamma'_{\neg P}$);
  then also $P$ could be characterised by a single formula (the negation of the formula for $\neg P$).
  Thus~\eqref{eq:finite:ac} fails for some ARSs. 
  The implication from left to right holds since $\Gamma'_{\neg P} \subseteq \Delta_{\neg P}$.
  Consequently, it is the implication from right to left which fails.
  So there exists an ARS $\aars$ such that $\aars \models \Gamma'_{\neg P}$ and $\aars$ has the property $P$.
  Then $\aars \models \Delta_{P}$ and $\aars \models \Gamma'_{P}$.
  Thus $\aars \models \Gamma'$.

  This is in contradiction to the compactness theorem,
  and hence our assumption must have been wrong.
  So $\neg P$ is not a gfop.
\end{proof}

\begin{theorem}\label{thm:fo:neg:ac}
  The property $\neg \AC$ is not a gfop.
\end{theorem}

\begin{proof}
  Follows from Lemma~\ref{lem:gfop:not:fop} and Theorem~\ref{thm:ac}.
\end{proof}

\section*{Undefinability via Finite Model Theory}

We will now reason about first-order definability using well-known techniques from the area of finite model theory.
In particular, we use \emph{Hanf's theorem} which is a central criterion for
establishing winning strategies in \emph{Ehrenfeucht-Fra\"{i}ss\'{e} games}.
For a general introduction to Ehrenfeucht-Fra\"{i}ss\'{e} games and Hanf's theorem we refer to~\cite{ebbi:flum:2005, libkin2013elements}.

\begin{definition}
  Two ARSs $\aars, \bars$ are \emph{elementarily equivalent} if 
  \begin{align*}
    \aars \models \set{\phi} \quad\iff\quad \bars \models \set{\phi} 
  \end{align*}
  for every first-order formula $\phi$ with equality and the predicate $\to$ (one-step rewriting).
\end{definition}

\newcommand{\sphere}[3]{\mathcal{N}^{#1}_{#2}(#3)}
\newcommand{\localbij}[1]{\rightleftarrows_{#1}}
\newcommand{\isomorphic}{\cong}
\begin{definition}
  Let $\aars = (A,\to)$ be an ARS and $r \in \nat$.
  \begin{enumerate}[(i)]
    \item 
      The \emph{degree} of an element $a \in A$ is the cardinality of the set $\set{b \mid a \to b \vee b \to a}$.
      We say that $\aars$ has \emph{finite degree} if the degree of every node is finite.
    \item 
      The \emph{distance} between nodes $a,b \in A$, denoted $\dist{a,b}$, is the length of the shortest path from $a$ to $b$,
      ignoring the direction of the arrows. 
      If no path exists, we stipulate that $\dist{a,b} = \infty$.
    \item 
      The \emph{$r$-neighbourhood $\sphere{\aars}{r}{a}$ of an element $a \in A$} is the restriction of $\aars$ to elements $\set{b \in A \mid \dist{a,b} \le r}$
      where $a$ is considered to be the root of the neighbourhood.
  \end{enumerate}
\end{definition}

Hanf's theorem uses the notion of \emph{Gaifman graphs} to define the distance $\dist{a,b}$.
In our setting of abstract rewrite systems, Gaifman graphs boil down to the underlying undirected graphs.

\begin{definition}
  Let $\aars = (A,\to_A)$ and $\bars = (B,\to_B)$ be ARSs and $r \in \nat$.
  \begin{enumerate}[(i)]
    \item 
      We write $\aars \isomorphic \bars$ if $\aars$ and $\bars$ are \emph{isomorphic}, that is, 
      there exists a bijection $f : A \to B$ such that $a \to_A b \iff f(a) \to_B f(b)$ for all $a,b \in A$.
    \item 
      We write $\aars \localbij{r} \bars$ if $\aars$ and $\bars$ are \emph{$r$-locally isomorphic}, that is, 
      there exists a bijection $f : A \to B$
      such that $\sphere{\aars}{r}{a} \isomorphic \sphere{\bars}{r}{f(a)}$ for every $a \in A$.
  \end{enumerate}
\end{definition}

Let $\card{S}$ denote the cardinality of a set $S$. 
\begin{lemma}\label{rem:local:isomorphism}
  We have $\aars \localbij{r} \bars$ if and only if 
  \begin{align}
    \card{\set{a \in A \mid \sphere{\aars}{r}{a} \isomorphic \mathcal{G}}}
    = \card{\set{b \in B \mid \sphere{\bars}{r}{b} \isomorphic \mathcal{G}}}
    \label{eq:g:count}
  \end{align}
  for every ARS $\mathcal{G}$.
\end{lemma}

\begin{proof}
  The lemma can be understood as follows.
  Assume that we are given sets $A, B$, a set of colours~$C$ and a colouring map 
  \begin{align*}
    c : (A \cup B) \to C
  \end{align*}
  A function $f : A \to B$ \emph{preserves colours} if $c(a) = c(f(a))$ for every $a \in A$.
  For $X \subseteq (A \cup B)$ and $d \in C$ we write $X|_d$ for the restriction of $X$ to colour $d$, that is:
  \begin{align*}
    X|_d = \set{x \in X \mid c(x) = d}
  \end{align*}
  Then, there exists a bijection $f : A \to B$ that preserves colours 
  if and only if, for every colour $d \in C$, 
  there exists a bijection $g : A|_d \to B|_d$:
  \begin{enumerate}
    \item 
      Given a colour preserving bijection $f : A \to B$ and a colour $d \in C$, 
      it follows that the restriction of $f$ to $A|_d$ is a bijection between $A|_d$ and $B|_d$. 
    \item 
      For every colour $d \in C$, let $g|_d : A|_d \to B|_d$ be a bijection.
      Define $f : A \to B$, for every $a \in A$, by $f(a) = g|_d(a)$ if $a \in A|_d$ for some $d \in C$.
      Note that $A$ is the disjoint union $A = \bigcup_{d \in C} A|_d$
      and likewise $B = \bigcup_{d \in C} B|_d$.
      It follows that $f$ is a bijection between $A$ and $B$. 
  \end{enumerate}
  In Equation~\ref{eq:g:count} we may think of $\mathcal{G}$ as the colour; it describes the $r$-neighbourhood of each node of that colour.
  To be precise, the colouring map is given by:
  \begin{align*}
    c(x) = \set{ \mathcal{G} \mid \text{$\sphere{\aars}{r}{x} \cong \mathcal{G}$ where $\mathcal{G}$ is some $r$-neighbourhood of $\aars$ or $\bars$  } }
  \end{align*}
  for $x \in A \cup B$.
\end{proof}

For the special case of ARSs, Hanf's theorem can be formulated as follows. 
\begin{theorem}[Hanf's theorem~\cite{ebbi:flum:2005,libkin2013elements}]\label{thm:hanf}
  ARSs $\aars$ and $\bars$ having finite degree are elementarily equivalent if $\aars \localbij{r} \bars$ for every $r \in \nat$.
\end{theorem}

\newcommand{\modelconfluence}{
  \begin{scope}
    \draw (0mm,0mm) node [n] (a) {};  
    \draw (-6mm,-6mm) node [n] (b0) {};
    \draw (6mm,-6mm) node [n] (c0) {};
    \draw (0mm,-12mm) node [n] (d) {};  

    \begin{scope}[->,cblue,shorten <= 1mm, shorten >= 1mm]
      \draw (a) -- (b0);
      \draw (a) -- (c0);
      \draw [<-] (b0) -- (d);
      \draw (c0) -- (d);
    \end{scope}
  \end{scope}

  \begin{scope}[xshift=18mm]
    \draw (0mm,0mm) node [n] (a) {};  
    \draw (-6mm,-6mm) node [n] (b0) {}
       ++ (0mm,-8mm) node [n] (b1) {};
    \draw (6mm,-6mm) node [n] (c0) {}
       ++ (0mm,-8mm) node [n] (c1) {};
    \draw (0mm,-20mm) node [n] (d) {};  

    \begin{scope}[->,cblue,shorten <= 1mm, shorten >= 1mm]
      \draw (a) -- (b0);
      \draw [<-] (b0) -- (b1);
      \draw [<-] (b1) -- (d);
      \draw (a) -- (c0);
      \draw (c0) -- (c1);
      \draw (c1) -- (d);
    \end{scope}
  \end{scope}

  \begin{scope}[xshift=36mm]
    \draw (0mm,0mm) node [n] (a) {};  
    \draw (-6mm,-6mm) node [n] (b0) {}
       ++ (0mm,-8mm) node [n] (b1) {}
       ++ (0mm,-8mm) node [n] (b2) {};
    \draw (6mm,-6mm) node [n] (c0) {}
       ++ (0mm,-8mm) node [n] (c1) {}
       ++ (0mm,-8mm) node [n] (c2) {};
    \draw (0mm,-28mm) node [n] (d) {};  

    \begin{scope}[->,cblue,shorten <= 1mm, shorten >= 1mm]
      \draw (a) -- (b0);
      \draw [<-] (b0) -- (b1);
      \draw [<-] (b1) -- (b2);
      \draw [<-] (b2) -- (d);
      \draw (a) -- (c0);
      \draw (c0) -- (c1);
      \draw (c1) -- (c2);
      \draw (c2) -- (d);
    \end{scope}
  \end{scope}

  \begin{scope}[xshift=49mm]
    \node at (0mm,-6mm) {$\cdots$};  
  \end{scope}
  
  \begin{scope}[xshift=62mm]
    \draw (0mm,0mm) node [n] (a) {};  
    \draw (-6mm,-6mm) node [n] (b0) {}
       ++ (0mm,-8mm) node [n] (b1) {}
       ++ (0mm,-9mm) node (b2) {$\vdots$}
       ++ (0mm,-10mm) node [n] (b3) {};
    \draw (6mm,-6mm) node [n] (c0) {}
       ++ (0mm,-8mm) node [n] (c1) {}
       ++ (0mm,-9mm) node (c2) {$\vdots$}
       ++ (0mm,-10mm) node [n] (c3) {}
       ++ (-6mm,-6mm) node [n] (d) {};

    \begin{scope}[->,cblue,shorten <= 1mm, shorten >= 1mm]
      \draw (a) -- (b0);
      \draw [<-] (b0) -- (b1);
      \draw [<-] [shorten >= 0mm] (b1) -- (b2);
      \draw [<-] (b2) -- (b3);
      \draw [<-] (b3) -- (d);
      \draw (a) -- (c0);
      \draw (c0) -- (c1);
      \draw [shorten >= 0mm] (c1) -- (c2);
      \draw (c2) -- (c3);
      \draw (c3) -- (d);
    \end{scope}
  \end{scope}

  \begin{scope}[xshift=75mm]
    \node at (0mm,-6mm) {$\cdots$};  
  \end{scope}
}
  
\begin{figure}[!ht]
  \centering
  \begin{tikzpicture}[default,node distance=5.5mm,inner sep=0mm,outer sep=0mm,n/.style={circle,minimum size=1mm,fill=black},scale=.8]
    \modelconfluence
    \node at (-10mm,0) {$\aars$};

    \begin{scope}[xshift=110mm]
      \draw (0mm,0mm) node [n] (a) {};  
      \draw (-6mm,-6mm) node [n] (b0) {}  
          ++(0mm,-8mm) node [n] (b1) {}
          ++(0mm,-8mm) node [n] (b2) {}
          ++(0mm,-8mm) node [n] (b3) {}
          ++(0mm,-9mm) node (b4) {$\vdots$};
      \draw (6mm,-6mm) node [n] (c0) {}  
          ++(0mm,-8mm) node [n] (c1) {}
          ++(0mm,-8mm) node [n] (c2) {}
          ++(0mm,-8mm) node [n] (c3) {}
          ++(0mm,-9mm) node (c4) {$\vdots$};
    
      \begin{scope}[->,cblue,shorten <= 1mm, shorten >= 1mm]
        \node [black] at (-10mm,0) {$\bars$};
        \draw (a) -- (b0);
        \draw [<-] (b0) -- (b1);
        \draw [<-] (b1) -- (b2);
        \draw [<-] (b2) -- (b3);
        \draw [<-] [shorten >= 0mm] (b3) -- (b4);
        \draw (a) -- (c0);
        \draw (c0) -- (c1);
        \draw (c1) -- (c2);
        \draw (c2) -- (c3);
        \draw [shorten >= 0mm] (c3) -- (c4);
      \end{scope}
    \end{scope}
  \end{tikzpicture}
  
  \caption{Confluence is not first-order definable: $\aars \localbij{r} (\aars \uplus \bars)$.
    Here $\aars$ consists of (the union of) the infinite sequence of finite graphs on the left, and 
    $\bars$ consists of the single infinite graph on the right.}
  \label{fig:cr:not:elementary}
\end{figure}

\begin{theorem}\label{thm:fo:cr:hanf}
  Confluence, local confluence, the normal form property and the cofinality property are not first-order definable.
  More precisely, $\CR$, $\WCR$, $\NF$, $\CP$, $\neg \CR$, $\neg \WCR$, $\neg \NF$ and $\neg \CP$ are not gfops.
\end{theorem}

\begin{proof}
  Consider the ARSs $\aars$ and $\bars$ in Figure~\ref{fig:cr:not:elementary}.
  It is easy to see that every $r$-neighbourhood of $\bars$ occurs $\aleph_0$ times in $\aars$,
  and at most $\aleph_0$ times in $\bars$.
  So it occurs equally often in $\aars$ as in $\aars \uplus \bars$, namely $\aleph_0$ times.
  Thus, by Lemma~\ref{rem:local:isomorphism}, we have $\aars \localbij{r} (\aars \uplus \bars)$ for every $r \in \nat$.
  Hence $\aars$ and $\aars \uplus \bars$ are elementarily equivalent by Theorem~\ref{thm:hanf},
  that is, they satisfy the same first-order formulas.
  Note that 
  \begin{align*}
    \aars &\models \CR, \WCR, \NF, \CP \\
    \aars \uplus \bars &\models \neg \CR, \neg \WCR, \neg \NF, \neg \CP
  \end{align*}
  As $\aars$ and $\aars \uplus \bars$ satisfy the same first-order formulas, 
  these properties are not gfops.
  For instance, assume that there is a set $\Phi$ of sentences characterising $\CR$,
  then 
  \begin{align*}
    \aars \models \CR \iff \aars \models \Phi \stackrel{\text{elementary equivalence}}{\iff} \aars \uplus \bars \models \Phi \iff \aars \uplus \bars \models \CR\;,
  \end{align*}
  but $\aars \models \CR$ and $\aars \uplus \bars \models \neg\CR$.
\end{proof}

\newcommand{\modelconfluencecolour}{
  \begin{scope}
    \draw (0mm,0mm) node [n,na] (a) {$\alpha$};  
    \draw (-6mm,-6mm) node [n,nb] (b0) {$\beta$};
    \draw (6mm,-6mm) node [n,ng] (c0) {$\gamma$};
    \draw (0mm,-12mm) node [n,nd] (d) {$\delta$};  

    \begin{scope}[->,cblue]
      \draw (a) -- (b0);
      \draw (a) -- (c0);
      \draw [<-] (b0) -- (d);
      \draw (c0) -- (d);
    \end{scope}
  \end{scope}

  \begin{scope}[xshift=21mm]
    \draw (0mm,0mm) node [n,ne] (a) {$\epsilon$};  
    \draw (-6mm,-6mm) node [n,nz] (b0) {$\zeta$}
       ++ (0mm,-8mm) node [n,ni] (b1) {$\iota$};
    \draw (6mm,-6mm) node [n,net] (c0) {$\eta$}
       ++ (0mm,-8mm) node [n,nk] (c1) {$\kappa$};
    \draw (0mm,-20mm) node [n,nk] (d) {$\kappa$};  

    \begin{scope}[->,cblue]
      \draw (a) -- (b0);
      \draw [<-] (b0) -- (b1);
      \draw [<-] (b1) -- (d);
      \draw (a) -- (c0);
      \draw (c0) -- (c1);
      \draw (c1) -- (d);
    \end{scope}
  \end{scope}

  \begin{scope}[xshift=42mm]
    \draw (0mm,0mm) node [n,ne] (a) {$\epsilon$};  
    \draw (-6mm,-6mm) node [n,nz] (b0) {$\zeta$}
       ++ (0mm,-8mm) node [n,ni] (b1) {$\iota$}
       ++ (0mm,-8mm) node [n,nk] (b2) {$\kappa$};
    \draw (6mm,-6mm) node [n,net] (c0) {$\eta$}
       ++ (0mm,-8mm) node [n,nk] (c1) {$\kappa$}
       ++ (0mm,-8mm) node [n,nk] (c2) {$\kappa$};
    \draw (0mm,-28mm) node [n,nk] (d) {$\kappa$};  

    \begin{scope}[->,cblue]
      \draw (a) -- (b0);
      \draw [<-] (b0) -- (b1);
      \draw [<-] (b1) -- (b2);
      \draw [<-] (b2) -- (d);
      \draw (a) -- (c0);
      \draw (c0) -- (c1);
      \draw (c1) -- (c2);
      \draw (c2) -- (d);
    \end{scope}
  \end{scope}

  \begin{scope}[xshift=57mm]
    \node at (0mm,-6mm) {$\cdots$};  
  \end{scope}
  
  \begin{scope}[xshift=71mm]
    \draw (0mm,0mm) node [n,ne] (a) {$\epsilon$};  
    \draw (-6mm,-6mm) node [n,nz] (b0) {$\zeta$}
       ++ (0mm,-8mm) node [n,ni] (b1) {$\iota$}
       ++ (0mm,-8mm) node [n,nk] (b2) {$\kappa$}
       ++ (0mm,-9mm) node (b3) {\raisebox{1.2mm}{$\vdots$}}
       ++ (0mm,-10mm) node [n,nk] (b4) {$\kappa$};
    \draw (6mm,-6mm) node [n,net] (c0) {$\eta$}
       ++ (0mm,-8mm) node [n,nk] (c1) {$\kappa$}
       ++ (0mm,-8mm) node [n,nk] (c2) {$\kappa$}
       ++ (0mm,-9mm) node (c3) {\raisebox{1.2mm}{$\vdots$}}
       ++ (0mm,-10mm) node [n,nk] (c4) {$\kappa$}
       ++ (-6mm,-6mm) node [n,nk] (d) {$\kappa$};

    \begin{scope}[->,cblue]
      \draw (a) -- (b0);
      \draw [<-] (b0) -- (b1);
      \draw [<-] [shorten >= 0mm] (b1) -- (b2);
      \draw [<-,shorten >= -1mm] (b2) -- (b3);
      \draw [<-] (b3) -- (b4);
      \draw [<-] (b4) -- (d);
      \draw (a) -- (c0);
      \draw (c0) -- (c1);
      \draw [shorten >= 0mm] (c1) -- (c2);
      \draw [shorten >= -1mm] (c2) -- (c3);
      \draw (c3) -- (c4);
      \draw (c4) -- (d);
    \end{scope}
  \end{scope}

  \begin{scope}[xshift=85mm]
    \node at (0mm,-6mm) {$\cdots$};  
  \end{scope}
}
  
\begin{figure}[!ht]
  \centering
  \begin{tikzpicture}[default,node distance=5.5mm,inner sep=0mm,outer sep=0mm,
    n/.style={circle,minimum size=5mm,inner sep=0mm,fill=cblue!20,scale=0.8},scale=.8,
    na/.style={fill=black!10},
    nb/.style={fill=black!17},
    ng/.style={fill=black!23},
    nd/.style={fill=black!30},
    ne/.style={fill=corange!30},
    net/.style={fill=purple!30},
    nz/.style={fill=cgreen!30},
    ni/.style={fill=cred!30},
    nk/.style={fill=cblue!30}]
    \modelconfluencecolour
    \node at (-10mm,0) {$\aars$};

    \begin{scope}[xshift=110mm]
      \draw (0mm,0mm) node [n,ne] (a) {$\epsilon$};  
      \draw (-6mm,-6mm) node [n,nz] (b0) {$\zeta$}  
          ++(0mm,-8mm) node [n,ni] (b1) {$\iota$}
          ++(0mm,-8mm) node [n,nk] (b2) {$\kappa$}
          ++(0mm,-8mm) node [n,nk] (b3) {$\kappa$}
          ++(0mm,-9mm) node (b4) {$\vdots$};
      \draw (6mm,-6mm) node [n,net] (c0) {$\eta$}  
          ++(0mm,-8mm) node [n,nk] (c1) {$\kappa$}
          ++(0mm,-8mm) node [n,nk] (c2) {$\kappa$}
          ++(0mm,-8mm) node [n,nk] (c3) {$\kappa$}
          ++(0mm,-9mm) node (c4) {$\vdots$};
    
      \begin{scope}[->,cblue]
        \node [black] at (-10mm,0) {$\bars$};
        \draw (a) -- (b0);
        \draw [<-] (b0) -- (b1);
        \draw [<-] (b1) -- (b2);
        \draw [<-] (b2) -- (b3);
        \draw [<-] [shorten >= 0mm] (b3) -- (b4);
        \draw (a) -- (c0);
        \draw (c0) -- (c1);
        \draw (c1) -- (c2);
        \draw (c2) -- (c3);
        \draw [shorten >= 0mm] (c3) -- (c4);
      \end{scope}
    \end{scope}
  \end{tikzpicture}
  
  \caption{The different $2$-neighbourhoods for Figure~\ref{fig:cr:not:elementary} discriminated using colours 
    (the Greek letters are for comparison in non-coloured renderings of this paper).
    They indicate the centre of the $2$-neighbourhood as distinguished element (the neighbourhoods are `rooted').
    This figure illustrates the proof of $2$-local isomorphism of $\aars$ and $\aars \uplus \bars$.}
  \label{fig:cr:colours}
\end{figure}

Figure~\ref{fig:cr:colours} visualises the colouring method described in the proof of Lemma~\ref{rem:local:isomorphism} when applied to Figure~\ref{fig:cr:not:elementary}.
So the different colours stand for different $2$-neighbourhoods.

\newcommand{\modelline}[2]{
  \begin{scope}[xshift=#2]
  \coordinate (c) at (0mm,0mm);  
  \node (n0) [n] at (c) {};
  \foreach \i in {1,...,#1} {
    \node (n\i) [n] at ([yshift=-8mm]c) {};
    \draw [->] (c) to (n\i);  
    \coordinate (c) at ([yshift=-8mm]c);
  } 
  \end{scope}
}

\begin{figure}[!ht]
  \centering
  \begin{tikzpicture}[default,node distance=5.5mm,inner sep=0mm,outer sep=0mm,n/.style={circle,minimum size=1mm,fill=black},scale=.8]
    \begin{scope}[->,cblue,shorten <= 1mm, shorten >= 1mm]
    \modelline{1}{0mm}
    \modelline{2}{10mm}
    \modelline{3}{20mm}
    \modelline{4}{30mm}
    \modelline{5}{40mm}
    \end{scope}
    \node at (-10mm,0) {$\aars$};
    \node at (50mm,-18mm) {$\cdots$};

    \begin{scope}[xshift=90mm]
      \begin{scope}[->,cblue,shorten <= 1mm, shorten >= 1mm]
      \modelline{5}{0mm}
      \end{scope}
      \node at (0mm,-45mm) {$\vdots$};
      \node at (-10mm,0) {$\bars$};
    \end{scope}
  \end{tikzpicture}
  
  \caption{Termination is not first-order definable: $\aars \localbij{r} (\aars \uplus \bars)$.}
  \label{fig:sn:not:elementary}
\end{figure}

\begin{theorem}\label{thm:fo:sn:hanf}
  Strong, weak normalisation and inductivity are not first-order definable.
  More precisely, the properties $\SN$, $\WN$, $\IND$, $\neg \SN$, $\neg \WN$ and $\neg \IND$
  are not gfops.
\end{theorem}

\begin{proof}
  Consider the ARSs $\aars$ and $\bars$ in Figure~\ref{fig:sn:not:elementary}.
  Note that 
  \begin{align*}
    \aars \models \SN, \WN, \IND &&
    \aars \uplus \bars \models \neg \SN, \neg \WN, \neg \IND
  \end{align*}
  As in the proof of Theorem~\ref{thm:fo:cr:hanf} it follows that $\aars$ and $\aars \uplus \bars$ are elementarily equivalent.
\end{proof}

\newcommand{\modellineinverse}[2]{
  \begin{scope}[xshift=#2]
  \coordinate (c) at (0mm,0mm);  
  \node (n0) [n] at (c) {};
  \foreach \i in {1,...,#1} {
    \node (n\i) [n] at ([yshift=-8mm]c) {};
    \draw [<-] (c) to (n\i);  
    \coordinate (c) at ([yshift=-8mm]c);
  } 
  \end{scope}
}

\begin{figure}[!ht]
  \centering
  \begin{tikzpicture}[default,node distance=5.5mm,inner sep=0mm,outer sep=0mm,n/.style={circle,minimum size=1mm,fill=black},scale=.8]
    \begin{scope}[<-,cblue,shorten <= 1mm, shorten >= 1mm]
    \modellineinverse{1}{0mm}
    \modellineinverse{2}{10mm}
    \modellineinverse{3}{20mm}
    \modellineinverse{4}{30mm}
    \modellineinverse{5}{40mm}
    \end{scope}
    \node at (-10mm,0) {$\aars$};
    \node at (50mm,-18mm) {$\cdots$};

    \begin{scope}[xshift=90mm]
      \begin{scope}[<-,cblue,shorten <= 1mm, shorten >= 1mm]
      \modellineinverse{5}{0mm}
      \end{scope}
      \node at (0mm,-45mm) {$\vdots$};
      \node at (-10mm,0) {$\bars$};
    \end{scope}
  \end{tikzpicture}
  
  \caption{Increasingness is not first-order definable: $\aars \localbij{r} (\aars \uplus \bars)$.}
  \label{fig:inc:not:elementary}
\end{figure}

\begin{theorem}\label{thm:fo:inc:hanf}
  The properties $\INC$ and $\neg \INC$ are not gfops.
\end{theorem}

\begin{proof}
  Consider the ARSs $\aars$ and $\bars$ in Figure~\ref{fig:inc:not:elementary}.
  Note that 
  \begin{align*}
    \aars \models \INC &&
    \aars \uplus \bars \models \neg \INC
  \end{align*}
  As in the proof of Theorem~\ref{thm:fo:cr:hanf} it follows that $\aars$ and $\aars \uplus \bars$ are elementarily equivalent.
\end{proof}

\newcommand{\modelstrongconfluence}{
  \begin{scope}[xshift=0mm]
    \draw (0mm,0mm) node [n] (a) {};  
    \draw (-6mm,-6mm) node [n] (b0) {}
       ++ (0mm,-8mm) node [n] (b1) {}
       ++ (0mm,-8mm) node [n] (b2) {};
    \draw (6mm,-6mm) node [n] (c0) {}
       ++ (0mm,-8mm) node [n] (c1) {}
       ++ (0mm,-8mm) node [n] (c2) {}
       ++ (-6mm,-6mm) node [n] (d) {};  

    \begin{scope}[->,cblue,shorten <= 1mm, shorten >= 1mm]
      \draw (a) -- (b0);
      \draw [<-] (b0) -- (b1);
      \draw [<-] (b1) -- (b2);
      \draw [<-] (b2) -- (d);
      \draw (a) -- (c0);
      \draw (c0) -- (c1);
      \draw (c1) -- (c2);
      \draw (c2) -- (d);

      \node (x) [n] at ($(b0)!.5!(c1)$) {};
      \draw (b0) -- (x);
      \draw (x) -- (c1);
    \end{scope}
  \end{scope}

  \begin{scope}[xshift=20mm]
    \draw (0mm,0mm) node [n] (a) {};  
    \draw (-6mm,-6mm) node [n] (b0) {}
       ++ (0mm,-8mm) node [n] (b1) {}
       ++ (0mm,-8mm) node [n] (b2) {}
       ++ (0mm,-8mm) node [n] (b3) {};
    \draw (6mm,-6mm) node [n] (c0) {}
       ++ (0mm,-8mm) node [n] (c1) {}
       ++ (0mm,-8mm) node [n] (c2) {}
       ++ (0mm,-8mm) node [n] (c3) {}
       ++ (-6mm,-6mm) node [n] (d) {};  

    \begin{scope}[->,cblue,shorten <= 1mm, shorten >= 1mm]
      \draw (a) -- (b0);
      \draw [<-] (b0) -- (b1);
      \draw [<-] (b1) -- (b2);
      \draw [<-] (b2) -- (b3);
      \draw [<-] (b3) -- (d);
      \draw (a) -- (c0);
      \draw (c0) -- (c1);
      \draw (c1) -- (c2);
      \draw (c2) -- (c3);
      \draw (c3) -- (d);

      \node (x) [n] at ($(b0)!.5!(c1)$) {};
      \draw (b0) -- (x);
      \draw (x) -- (c1);
    \end{scope}
  \end{scope}

  \begin{scope}[xshift=40mm]
    \draw (0mm,0mm) node [n] (a) {};  
    \draw (-6mm,-6mm) node [n] (b0) {}
       ++ (0mm,-8mm) node [n] (b1) {}
       ++ (0mm,-8mm) node [n] (b2) {}
       ++ (0mm,-8mm) node [n] (b3) {}
       ++ (0mm,-8mm) node [n] (b4) {};
    \draw (6mm,-6mm) node [n] (c0) {}
       ++ (0mm,-8mm) node [n] (c1) {}
       ++ (0mm,-8mm) node [n] (c2) {}
       ++ (0mm,-8mm) node [n] (c3) {}
       ++ (0mm,-8mm) node [n] (c4) {}
       ++ (-6mm,-6mm) node [n] (d) {};  

    \begin{scope}[->,cblue,shorten <= 1mm, shorten >= 1mm]
      \draw (a) -- (b0);
      \draw [<-] (b0) -- (b1);
      \draw [<-] (b1) -- (b2);
      \draw [<-] (b2) -- (b3);
      \draw [<-] (b3) -- (b4);
      \draw [<-] (b4) -- (d);
      \draw (a) -- (c0);
      \draw (c0) -- (c1);
      \draw (c1) -- (c2);
      \draw (c2) -- (c3);
      \draw (c3) -- (c4);
      \draw (c4) -- (d);

      \node (x) [n] at ($(b0)!.5!(c1)$) {};
      \draw (b0) -- (x);
      \draw (x) -- (c1);
    \end{scope}
  \end{scope}

  \begin{scope}[xshift=54mm]
    \node at (0mm,-22mm) {$\cdots$};  
  \end{scope}
}
\begin{figure}[!ht]
  \centering
  \begin{tikzpicture}[default,node distance=5.5mm,inner sep=0mm,outer sep=0mm,n/.style={circle,minimum size=1mm,fill=black},scale=.8]
    \modelstrongconfluence
    \node at (-10mm,0) {$\aars$};

    \begin{scope}[xshift=80mm]
      \draw (0mm,0mm) node [n] (a) {};  
      \draw (-6mm,-6mm) node [n] (b0) {}  
          ++(0mm,-8mm) node [n] (b1) {}
          ++(0mm,-8mm) node [n] (b2) {}
          ++(0mm,-8mm) node [n] (b3) {}
          ++(0mm,-9mm) node (b4) {$\vdots$};
      \draw (6mm,-6mm) node [n] (c0) {}  
          ++(0mm,-8mm) node [n] (c1) {}
          ++(0mm,-8mm) node [n] (c2) {}
          ++(0mm,-8mm) node [n] (c3) {}
          ++(0mm,-9mm) node (c4) {$\vdots$};
    
      \begin{scope}[->,cblue,shorten <= 1mm, shorten >= 1mm]
        \node [black] at (-10mm,0) {$\bars$};
        \draw (a) -- (b0);
        \draw [<-] (b0) -- (b1);
        \draw [<-] (b1) -- (b2);
        \draw [<-] (b2) -- (b3);
        \draw [<-] [shorten >= 0mm] (b3) -- (b4);
        \draw (a) -- (c0);
        \draw (c0) -- (c1);
        \draw (c1) -- (c2);
        \draw (c2) -- (c3);
        \draw [shorten >= 0mm] (c3) -- (c4);

        \node (x) [n] at ($(b0)!.5!(c1)$) {};
        \draw (b0) -- (x);
        \draw (x) -- (c1);
      \end{scope}
    \end{scope}
  \end{tikzpicture}
  
  \caption{Strong confluence is not first-order definable: $\aars \localbij{r} (\aars \uplus \bars)$.}
  \label{fig:sc:not:elementary}
\end{figure}

\begin{theorem}\label{thm:fo:sc:hanf}
  Strong confluence is not first-order definable.
  More precisely, the properties $\SC$ and $\neg \SC$ are not gfops.
\end{theorem}

\begin{proof}
  Consider the ARSs $\aars$ and $\bars$ in Figure~\ref{fig:sc:not:elementary}.
  Note that 
  \begin{align*}
    \aars \models \SC &&
    \aars \uplus \bars \models \neg \SC
  \end{align*}
  To see that $\aars \models \SC$, consider one component of $\aars$.
  In this component there is one peak, say $b \lstep a \step c$, where $b$ is the element displayed on the left and $c$ the element on the right.
  Then the peak
  \begin{enumerate}
    \item 
      $b \lstep a \step c$ can be joined by $b \step^2 \cdot \lstep c$, and
    \item   
      $c \lstep a \step b$ can be joined by $b \lsteps c$.
  \end{enumerate}
  As in the proof of Theorem~\ref{thm:fo:cr:hanf} it follows that $\aars$ and $\aars \uplus \bars$ are elementarily equivalent.
\end{proof}

%% file: two.tex
In this section we show that two labels suffice 
for proving confluence using decreasing diagrams
for any abstract rewrite system having the cofinality property.
We start by introducing the decreasing diagrams technique.

\begin{notation}
  For an indexed ARS $\aars = \indars$ and a relation ${<} \subseteq I\times I$, we define
  \begin{align*}
    {\to} \;&=\; \textstyle{\bigcup_{\alpha\in I} \to_\alpha} &
    {\to_{< \beta}} \;&=\; \textstyle{\bigcup_{\alpha < \beta} \to_\alpha}
  \end{align*}
  Moreover, we use $\to_{{< \alpha} \cup {< \beta}}$ as shorthand for $\to_{< \alpha} \cup \to_{< \beta}$.
\end{notation}

\begin{definition}[Decreasing Church--Rosser~\cite{oost:1994b}]\label{def:decreasing:diagrams}
  An ARS $\aars = (A,\rightsquigarrow)$ is called \emph{decreasing Church--Rosser (\DCR)}
  if there exists an ARS $\bars = \indars$ indexed by a well-founded partial order $(I,<)$
  such that ${\rightsquigarrow} = {\to}$ and
  every peak $c \lstep_\beta a \to_\alpha b$ can be joined 
  by reductions of the form shown in Figure~\ref{fig:decreasing:diagram}.\footnote{%
    Van Oostrom~\cite{oost:2008} generalises the shape of the decreasing elementary diagrams by allowing the joining reductions to be conversions.
    This can be helpful to find suitable elementary diagrams.
    However, if there are conversions then we can always obtain joining reductions by diagram tiling.
    So a system is locally decreasing with respect to conversions if and only if it is locally decreasing with respect to reductions (using the same labelling of the steps).
  }
\end{definition}

\begin{figure}[!ht]
  \centering
  \begin{tikzpicture}[default,thick,baseline=0ex,every node/.style={rectangle},inner sep=1mm]
    \node (a) {$a$};
    \node (b) [right of=a,node distance=35mm] {$b$};
    \node (c) [below of=a,node distance=35mm] {$c$};
    \node (d) [right of=c,node distance=35mm] {$d$};
    \node (b') [smallCircle] at ($(b)!.33!(d)$) {};
    \node (b'') [smallCircle] at ($(b)!.66!(d)$) {};
    \node (c') [smallCircle] at ($(c)!.33!(d)$) {};
    \node (c'') [smallCircle] at ($(c)!.66!(d)$) {};
    \draw [->] (a) -- node [above] {$\alpha$} (b);
    \draw [->] (a) -- node [left] {$\beta$} (c);
    \begin{scope}[exi]
    \draw [->>] (b) -- node [right] {$< \alpha$} (b');
    \draw [->] (b') -- node [right] {$\beta$} node [left] {$\equiv$} (b'');
    \draw [->>] (b'') -- node [right,align=center] {$<\alpha$ $\cup$ $<\beta$} (d);
    \draw [->>] (c) -- node [below,yshift=-.4mm] {$< \beta$} (c');
    \draw [->] (c') -- node [below,yshift=-.4mm,align=center] {$\alpha$} node [above] {$\equiv$} (c'');
    \draw [->>] (c'') -- node [below,yshift=-.4mm,align=center] {$<\alpha$\\[-.45ex]\ $\cup$ $<\beta$} (d);
    \end{scope}
    
%
  \end{tikzpicture}
  \vspace{-1ex}
  \caption{Decreasing elementary diagram.
  }
  \label{fig:decreasing:diagram}
\end{figure}

The following is the main theorem of decreasing diagrams.
\begin{theorem}[Decreasing Diagrams -- De Bruijn~\cite{brui:1978} \& Van Oostrom~\cite{oost:1994b}]\label{thm:decreasing:diagrams}
  If an ARS is decreasing Church--Rosser,
  then it is confluent. In other words $\DCR \implies \CR$.
  \qed
\end{theorem}

As already suggested in the introduction, we introduce classes $\DCR_\alpha$
restricting the well-founded order $(I,<)$ in Definition~\ref{def:decreasing:diagrams} to the ordinal $\alpha$.

\begin{definition}\label{def:dcr:alpha}
  For ordinals $\alpha$,
  let $\DCR_\alpha$ denote the class of ARSs $\aars$
  that are decreasing Church--Rosser (Definition~\ref{def:decreasing:diagrams}) with label set $\{\,\beta \mid \beta < \alpha\,\}$
  ordered by the usual order~$<$ on ordinals.
  We say that $\aars$
  has the property $\DCR_\alpha$, denoted $\DCR_\alpha(\aars)$, if $\aars \in \DCR_\alpha$.
\end{definition}

\begin{theorem}\label{thm:fo:dcr:hanf}
  For $\alpha \ge 2$, $\DCR_\alpha$, $\neg \DCR_\alpha$, $\DCR$ and $\neg \DCR$ are not gfops.
\end{theorem}

\begin{proof}
  Follows by an extension of the proof for Theorem~\ref{thm:fo:cr:hanf},
  noting that $\aars$ admits a decreasing labelling with $2$ labels.
  So $\aars \models \DCR_\alpha$ and $\aars \uplus \bars \models \neg \DCR_\alpha$.
\end{proof}

Note that $\DCR_1$ is equivalent to the diamond property for the reflexive closure of the rewrite relation,
and thus is first-order definable.

The remainder of this section is devoted to the proof that every system with the cofinality property is $\DCR_2$.
Put differently, it suffices to label steps with $I = \{\,0,1\,\}$.
Let $\aars = (A,\to)$ be an ARS having the cofinality property.
Note that, for defining the labelling, 
we can consider connected components with respect to $\lrstep^*$ separately.
Thus assume that $\aars$ consists of a single connected component,
that is, for every $a,b \in A$ we have $a \lrstep^* b$.
By the cofinality property, which implies confluence, and Lemma~\ref{lem:cofinal}
there exists a rewrite sequence
\begin{align*}
  m_0 \to m_1 \to m_2 \to m_3 \to \cdots
\end{align*}
that is cofinal in $\aars$; we call this rewrite sequence the \emph{main road}.
Without loss of generality we may assume that the main road is acyclic,
that is, $m_i \not\equiv m_j$ whenever $i \ne j$.
(We can eliminate loops without harming the cofinality property. Note that the main road is allowed to be finite.)

The idea of labelling the steps in $\aars$ is as follows.
For every node $a \in A$, 
we label precisely one of the outgoing edges with $0$ and all others with $1$.
The edge labelled with $0$ must be part of a shortest path from $a$ to the main road.
For the case that $a$ lies on the main road, the step labelled $0$ must be the step on the main road. 
This is illustrated in Figure~\ref{fig:example_labelling}.

\begin{figure}[!ht]
  \centering
  \begin{tikzpicture}[default,label/.style={scale=.8,color=black},n/.style={}]
    \draw (0mm,0mm) node (m0) {$m_0$}
        ++(10:14mm) node (m1) {$m_1$}
        ++(20:14mm) node (m2) {$m_2$}
        ++(30:14mm) node (m3) {$m_3$}
        ++(-10:14mm) node (m4) {$m_4$}
        ++(-20:14mm) node (m5) {$m_5$}
        ++(-5:14mm) node (m6) {$\cdots$}
        ;
    
    \begin{scope}[->, cred, line width=0.5mm]
      \draw (m0) to node [above,label] {$0$} (m1);
      \draw (m1) to node [above,label] {$0$} (m2);
      \draw (m2) to node [above left,label] {$0$} (m3);
      \draw (m3) to node [above,label] {$0$} (m4);
      \draw (m4) to node [above,label] {$0$} (m5);
      \draw (m5) to node [above,label] {$0$} (m6);
    \end{scope}

    \draw (m2) ++(80:14mm) node (n0) [n] {$n_0$}
               ++(60:14mm) node (n1) [n] {$n_1$}
               ++(10:14mm) node (n2) [n] {$n_2$}
               ++(-30:14mm) node (n3) [n] {$n_3$};

    \begin{scope}[->]
      \draw [lightblue,dashed] (m2) to node [left,label] {$1$} (n0);
      \draw [cgreen] (n0) to node [left,label] {$0$} (n1);
      \draw [lightblue,dashed] (n1) to node [above,label] {$1$} (n2);
      \draw [cgreen] (n2) to node [above,label] {$0$} (n3);
      \draw [cgreen] (n3) to node [right,label] {$0$} (m5);
    \end{scope}
    
    \draw (n1) ++(170:14mm) node (n4) [n] {$n_4$}
               ++(-160:14mm) node (n5) [n] {$n_5$}
               ++(-130:14mm) node (n6) [n] {$n_6$};

    \begin{scope}[->]
      \draw [lightblue,dashed] (n1) to node [above,label] {$1$} (n4);
      \draw [cgreen] (n4) to node [above,label] {$0$} (n5);
      \draw [lightblue,dashed] (n5) to node [above left,label] {$1$} (n6);
      \draw [cgreen] (n6) to node [left,label] {$0$} (m0);
      \draw [cgreen] (n1) to[bend left=15] node [below,label] {$0$} (n5);
      \draw [cgreen] (n5) to[bend left=15] node [right,label] {$0$} (m1);
      \draw [lightblue,dashed] (n4) to[bend left=-60] node [above left,label] {$1$} (n6);
    \end{scope}

    \draw (n3) ++(-160:14mm) node (n7) [n] {$n_7$};

    \begin{scope}[->]
      \draw [lightblue,dashed] (n2) to node [left,label] {$1$} (n7);
      \draw [lightblue,dashed] (n3) to node [above,label] {$1$} (n7);
      \draw [cgreen] (n7) to node [left,label] {$0$} (m4);

      \draw [lightblue,dashed] (m2) to[bend left=-20] node [below,label] {$1$} (m5);
    \end{scope}
    
    \draw [->,cred, line width=0.5mm] (8.5cm,27mm) to node [at end,anchor=west,xshift=1mm,yshift=.5mm] {main road} ++(7mm,0mm);
    \draw [->,cgreen] (8.5cm,20mm) to node [at end,anchor=west,xshift=1mm,yshift=.25mm,cdarkgreen] {minimising} ++(7mm,0mm);
    \draw [->,lightblue,dashed] (8.5cm,13mm) to node [at end,anchor=west,xshift=1mm,yshift=.5mm] {non-minimising} ++(7mm,0mm);
  \end{tikzpicture}\vspace{-2ex}
  \caption{Example labelling.}
  \label{fig:example_labelling}
\end{figure}

Note that there is a choice about which edge to label with $0$
whenever there are multiple outgoing edges that all start a shortest path to the main road. 
To resolve this choice, the following definition assumes a well-order $<$ on the universe $A$, 
whose existence is guaranteed by the well-ordering theorem.
Then, whenever there is a choice, we choose the edge
for which the target is minimal in this order.

\begin{remark}
  Recall that the Axiom of Choice is equivalent to the well-ordering theorem.
  In many practical cases, however, the existence of such a well-order
  does not require the Axiom of Choice.
  If the universe is countable, then such a well-order can be derived
  directly from the surjective counting function $f : \nat \to A$.
\end{remark}

In the following definition we follow the proof in~\cite[Proposition~14.2.30, p.\ 766]{tere:2003},
employing the notion of a cofinal sequence and the rewrite distance from a point to this sequence.
While the proof in~\cite{tere:2003} labels steps by their distance to the target node, 
we need a more sophisticated labelling.

\begin{definition}
  Let $\aars = (A,\to)$ be an ARS and $M : m_0 \to m_1 \to m_2 \to \cdots$ be a finite or infinite rewrite sequence in $\aars$.
  For $a,b \in A$, we write 
  \begin{enumerate}[label=(\roman*)]
    \item $a \in M$ if $a \equiv m_i$ for some $i \ge 0$, and
    \item $(a \to b) \in M$ if $a \equiv m_i$ and $b \equiv m_{i+1}$ for some $i \ge 0$.
  \end{enumerate}
  If $M$ is cofinal in $\aars$, 
  we define the \emph{distance} $\distm{a}{M}$ 
  as the least natural number $n \in \nat$ such that $a \to^n m$ for some $m \in M$.
  If $M$ is clear from the context, we write $\dist{a}$ for $\distm{a}{M}$.
\end{definition}

\begin{definition}[Labelling with two labels]\label{def:labelling}
  Let $\aars = (A,\to)$ be an ARS equipped with a well-order $<$ on $A$
  such that there exists a cofinal reduction $M : m_0 \to m_1 \to m_2 \to \cdots$
  that is acyclic (that is, for all $i < j$, $m_i \not\equiv m_j$).
   
  We say that a step $a \to b$ is 
  \begin{enumerate}[label=(\roman*)]
    \item 
      \emph{on the main road} 
      if $(a \to b) \in M$;
    \item 
      \emph{minimising} if $\dist{a} = \dist{b}+1$ and
      $b' \geq b$ for every $a \to b'$ with $\dist{b'} = \dist{b}$.
  \end{enumerate}
  We define an indexed ARS $\lab{\aars} = (A,\fam{\to_i}_{i \in I})$ where $I = \{\,0,1\,\}$ as follows: 
  \begin{align*}
    a \to_0 b \;&\iff\; \text{$a \to b$ and this step is on the main road or minimising}\\
    a \to_1 b \;&\iff\; \text{$a \to b$ and this step is not on the main road and not minimising}
  \end{align*}
  for every $a,b \in A$.
\end{definition}

\begin{lemma}\label{lem:labelars}
  Let $\aars = (A,\to)$ be an ARS with a cofinal rewrite sequence $M : m_0 \to m_1 \to \cdots$
  that is acyclic. Furthermore,
  let $<$ be a well-order over $A$.  
  Then for $\lab{\aars} = (A,\to_0,\to_1)$ we have:
  \begin{enumerate}[label=(\roman*)]
    \item\label{it:union} 
      ${\to} \;=\; {\to_0 \cup \to_1}$\;;
    \item\label{it:main:join} 
      for every $a,b \in M$ we have $a \steps_0 \cdot \lsteps_0 b$\;;
    \item\label{it:zero:at:most} 
      for every $a \in A$, there is at most one $b \in A$ such that $a \to_0 b$\;;
    \item\label{it:zero:decrease} 
      for every $a \notin M$, there exists $b \in A$ with $a \to_0 b$ and $\dist{a} > \dist{b}$\;;
    \item\label{it:zero:to:main} 
      for every $a \in A$, there exists $m \in M$ such that $a \steps_0 m$\;;
    \item\label{it:peak} 
      every peak $c \lstep_\beta a \step_\alpha b$ can be joined as in Figure~\ref{fig:decreasing:diagram},
      and, explicitly for labels $\{0,1\}$, as in Figure~\ref{fig:decreasing:diagram:01}.
  \end{enumerate}
\end{lemma}

\begin{proof}
  Properties \ref{it:union} and \ref{it:main:join} follow from the definitions.
  
  For~\ref{it:zero:at:most} assume that $b \lstep_0 a \step_0 c$. We show that $b \equiv c$.
  The steps $a \to b$ and $a \to c$ are either minimising or on the main road.
  We distinguish cases $a \in M$ and $a \not\in M$:
  \begin{enumerate}[label=(\roman*)]
    \item 
      Assume that $a \in M$.
      Then $d(a) = 0$, and thus neither $a \step b$ nor $a \step c$ is a minimising step.
      Hence $(a \step b) \in M$ and $(a \step c) \in M$. 
      Since $M$ is acyclic, we get $b \equiv c$.
    \item 
      If $a \notin M$, both steps $a \to b$ and $a \to c$ must be minimising.
      If $d(b) \neq d(c)$, then we have either $\dist{a} \ne \dist{b}+1$ or $\dist{a} \ne \dist{c}+1$, 
      contradicting minimisation.
      Thus $d(b) = d(c)$.
      Then by minimisation we have $b \geq c$ and $c \geq b$, from which we obtain $b \equiv c$.
  \end{enumerate}
  For~\ref{it:zero:decrease}, consider an element $a \notin M$.
  Let $B = \{b' \mid a \to b' \wedge \dist{a} = \dist{b'}+1 \}$.
  By definition of the distance $\dist{\cdot}$, $B \ne \varnothing$.
  Define $b$ as the least element of $B$ in the well-order $<$ on~$A$.
  It follows that $a \to b$ is a minimisation step. Hence $a \to_0 b$ and $\dist{a} > \dist{b}$.
  Property~\ref{it:zero:to:main} follows directly from \ref{it:zero:decrease}
  using induction on the distance.
  
  For~\ref{it:peak}, consider a peak $c \lstep_\beta a \step_\alpha b$.
  If $b \equiv c$, then the joining reductions are empty steps.
  Thus assume that $b \not\equiv c$.
  By~\ref{it:zero:at:most} we have either $\alpha = 1$ or $\beta = 1$.
  By~\ref{it:zero:to:main} there exist $m_b, m_c \in M$ such that
  $b \steps_0 m_b$ and $c \steps_0 m_c$.  
  By~\ref{it:main:join} we have $m_b \steps_0 \cdot \lsteps_0 m_c$.
  Hence $b \steps_0 \cdot \lsteps_0 c$.
  These joining reductions are of the form required by Figure~\ref{fig:decreasing:diagram} 
  since ${\steps_0} = {\steps_{<\alpha \cup <\beta}}$.
\end{proof}

\begin{theorem}\label{thm:01}
  If an ARS $\aars = (A,\to)$ satisfies the cofinality property,
  then there exists an indexed ARS $(A, (\step_\alpha)_{\alpha \in \{0,1\}})$
  such that ${\to} = {\to_0 \cup \to_1}$ and
  every peak $c \lstep_\beta a \step_\alpha b$ can be joined according
  to the elementary decreasing diagram in Figure~\ref{fig:decreasing:diagram},
  and, explicitly for labels $\{0,1\}$, as in Figure~\ref{fig:decreasing:diagram:01}.
\end{theorem}

\begin{proof}
  It suffices to consider a connected component of $\aars$.
  Let $\bars = (B,\to)$ be a connected component of $\aars$:
  we have $a \lrstep^* b$ for all $a,b\in B$.
  By the cofinality property and Lemma~\ref{lem:cofinal}, there exists a cofinal reduction $m_0 \to m_1 \to \cdots$ in $\bars$. By the well-ordering theorem, there exists a well-order $<$ over $B$.
  Then $\bars$ has the required properties by Lemma~\ref{lem:labelars}\ref{it:peak}.
\end{proof}

\begin{corollary}
  $\mathit{DCR}_2$ is a complete method for proving confluence of countable ARSs.
\end{corollary}
\begin{proof}
  Immediate from Theorems \ref{thm:klop} and \ref{thm:01}.
\end{proof}

Theorem~\ref{thm:01} also holds for De Bruijn's weak diamond property.
Note the following caveat: when restricting the index set $I$ to a single label, 
the decreasing diagram technique is equivalent to
${\lstep \cdot \step} \subseteq {\step^\equiv \cdot \lstep^\equiv}$,
i.e.\ the \emph{diamond property} for $\to \cup \equiv$, 
while the weak diamond property with one label is equivalent to 
\emph{strong confluence}
${\lstep \cdot \step} \subseteq {\step^\equiv \cdot \lsteps}$.

\begin{figure}[!ht]
  \centering
  \begin{tikzpicture}[default,node distance=25mm,nodes={rectangle,inner sep=1.5mm}]
    \begin{scope}
    \node (a) [smallCircle] {};
    \node (b) [smallCircle,right of=a] {};
    \node (c) [smallCircle,below of=a] {};
    \node (d) [smallCircle,below of=b] {};
    
    \draw [->] (a) -- node [above] {$0$} (b);
    \draw [->] (a) -- node [left] {$0$} (c);
    \begin{scope}[exi]
    \draw [->] (b) -- node [right] {$0$} node [left] {$\equiv$} (d);
    \draw [->] (c) -- node [below] {$0$} node [above] {$\equiv$} (d);
    \end{scope}
    \end{scope}

    \begin{scope}[xshift=45mm]
    \node (a) [smallCircle] {};
    \node (b) [smallCircle,right of=a] {};
    \node (c) [smallCircle,below of=a] {};
    \node (d) [smallCircle,below of=b] {};
    \node (b1) [smallCircle] at ($(b)!.5!(d)$) {};
    
    \draw [->] (a) -- node [above] {$0$} (b);
    \draw [->] (a) -- node [left] {$1$} (c);
    \begin{scope}[exi]
    \draw [->] (b) -- node [right] {$1$} node [left] {$\equiv$} (b1);
    \draw [->>] (b1) -- node [right] {$0$} (d);
    \draw [->>] (c) -- node [below] {$0$} (d);
    \end{scope}
    \end{scope}

    \begin{scope}[xshift=90mm]
    \node (a) [smallCircle] {};
    \node (b) [smallCircle,right of=a] {};
    \node (c) [smallCircle,below of=a] {};
    \node (d) [smallCircle,below of=b] {};
    \node (b1) [smallCircle] at ($(b)!.33!(d)$) {};
    \node (b2) [smallCircle] at ($(b)!.66!(d)$) {};
    \node (c1) [smallCircle] at ($(c)!.33!(d)$) {};
    \node (c2) [smallCircle] at ($(c)!.66!(d)$) {};
    
    \draw [->] (a) -- node [above] {$1$} (b);
    \draw [->] (a) -- node [left] {$1$} (c);
    \begin{scope}[exi]
    \draw [->>] (b) -- node [right] {$0$} (b1);
    \draw [->] (b1) -- node [right] {$1$} node [left] {$\equiv$} (b2);
    \draw [->>] (b2) -- node [right] {$0$} (d);
    \draw [->>] (c) -- node [below] {$0$} (c1);
    \draw [->] (c1) -- node [below] {$1$} node [above] {$\equiv$} (c2);
    \draw [->>] (c2) -- node [below] {$0$} (d);
    \end{scope}
    \end{scope}
  \end{tikzpicture}\vspace{-1.5ex}
  \caption{Decreasing diagrams with labels $0$ and $1$ where $0 < 1$.}
  \label{fig:decreasing:diagram:01}
\end{figure}

The property $\DCR_2$ is given implicitly by the decreasing diagrams as in Figure~\ref{fig:decreasing:diagram},
but it is also instructive to give explicitly the elementary reduction diagrams 
making up the property $\DCR_2$. These are shown in Figure~\ref{fig:decreasing:diagram:01}.
Note that the $1$-steps do not split in the diagram construction,
i.e., they cross over in at most one copy.
This facilitates a simple proof of confluence.

Actually, from our proof it follows that the joining reductions can be required to only contain steps with label $0$.
Thus even the simple shape of diagrams shown in Figure~\ref{fig:decreasing:diagram:01:simple} 
is complete for proving confluence of systems having the cofinality property.
Here the $1$-steps do not cross over at all!
Note that while this set of elementary diagrams has a trivial proof of confluence, 
the work to prove $\DCR_2 \implies \CR$ from the original elementary diagrams as in Figure~\ref{fig:decreasing:diagram:01}, 
consists in showing from our earlier construction that it actually suffices to join by using only 0's. 
\begin{figure}[h]
  \centering
  \begin{tikzpicture}[default,node distance=15mm,nodes={rectangle,inner sep=1.5mm}]
    \begin{scope}
    \node (a) [smallCircle] {};
    \node (b) [smallCircle,right of=a] {};
    \node (c) [smallCircle,below of=a] {};
    \node (d) [smallCircle,below of=b] {};
    
    \draw [->] (a) -- node [above] {$0$} (b);
    \draw [->] (a) -- node [left] {$0$} (c);
    \begin{scope}[exi]
    \draw [-] (b) -- node [right] {$\equiv$} (d);
    \draw [-] (c) -- node [below] {$\equiv$} (d);
    \end{scope}
    \end{scope}

    \begin{scope}[xshift=30mm]
    \node (a) [smallCircle] {};
    \node (b) [smallCircle,right of=a] {};
    \node (c) [smallCircle,below of=a] {};
    \node (d) [smallCircle,below of=b] {};
    
    \draw [->] (a) -- node [above] {$0$} (b);
    \draw [->] (a) -- node [left] {$1$} (c);
    \begin{scope}[exi]
    \draw [->>] (b) -- node [right] {$0$} (d);
    \draw [->>] (c) -- node [below] {$0$} (d);
    \end{scope}
    \end{scope}

    \begin{scope}[xshift=60mm]
    \node (a) [smallCircle] {};
    \node (b) [smallCircle,right of=a] {};
    \node (c) [smallCircle,below of=a] {};
    \node (d) [smallCircle,below of=b] {};
    
    \draw [->] (a) -- node [above] {$1$} (b);
    \draw [->] (a) -- node [left] {$1$} (c);
    \begin{scope}[exi]
    \draw [->>] (b) -- node [right] {$0$} (d);
    \draw [->>] (c) -- node [below] {$0$} (d);
    \end{scope}
    \end{scope}
  \end{tikzpicture}\vspace{-1.5ex}
  \caption{A simple set of diagrams that is complete for confluence of countable systems.}
  \label{fig:decreasing:diagram:01:simple}
\end{figure}

\begin{remark}
  We note a certain similarity between the notion of a decreasing diagram
  based on labels $\{\,0,1\,\}$ with $0 < 1$ and the classical `requests' lemma
  of J. Staples~\cite[Exercise 2.08.5, p.\ 9]{klop:1992,tere:2003}.
  In $\aars = (A,\to_1,\to_2)$ define: $\to_1$ requests $\to_2$ if
  \begin{center}
    \begin{tikzpicture}[default,node distance=18mm,nodes={rectangle,inner sep=1.5mm}]
      \node (a) [smallCircle] {};
      \node (b) [smallCircle,right of=a] {};
      \node (c) [smallCircle,below of=a] {};
      \node (d) [smallCircle,below of=b] {};
      \node (b1) [smallCircle] at ($(b)!.5!(d)$) {};
      
      \draw [->>] (a) -- node [above] {$2$} (b);
      \draw [->>] (a) -- node [left] {$1$} (c);
      \begin{scope}[exi]
      \draw [->>] (b) -- node [right] {$1$} (b1);
      \draw [->>] (b1) -- node [right] {$2$} (d);
      \draw [->>] (c) -- node [below] {$2$} (d);
      \end{scope}
    \end{tikzpicture}
  \end{center}
  If in addition $\to_1$ and $\to_2$ are confluent, then ${\to_{1,2}} = {\to_1 \cup \to_2}$ is confluent.
  
  The requests lemma states that the `dominant' reduction $\steps_1$ needs the `support'
  of the secondary reduction $\steps_2$ for making the divergence $\lsteps_1 \cdot \steps_2$ convergent.
  Similarly for the property $\DCR_2$, the dominant reduction $\to_1$ needs support by $\steps_0$
  for making the divergence $\lstep_1 \cdot \step_0$ convergent.
  However, the requests lemma employs $\steps$, not $\step$. 
\end{remark}


%% file: commutation.tex
The decreasing diagram technique can also be used for proving commutation, see~\cite{oost:1994b}.
It turns out that the situation for commutation stands in sharp contrast to that for confluence.
For commutation the hierarchy does not collapse.
In particular, we show that, for every $n \le \omega$, 
decreasing diagrams for commutation with $n$ labels is \emph{strictly} stronger than decreasing diagrams with less than $n$ labels. 

The elementary decreasing diagram for commutation is shown in Figure~\ref{fig:decreasing:diagram:commutation},
which is very similar to Figure~\ref{fig:decreasing:diagram}, but now refers to two `basis' relations $\ver$, $\hor$.

\begin{definition}[Decreasing Commutation]\label{def:decreasing:diagrams:commutation}
  An ARS $\aars = (A,{\ver},{\hor})$ is called \emph{decreasing commuting (\DC)} 
  if there is an ARS $\bars = (A,\fam{\ver_\alpha}_{\alpha \in I},\fam{\hor_\alpha}_{\alpha \in I})$ 
  indexed by a well-founded partial order $(I,<)$
  such that 
  ${\ver_\aars} = {\ver_\bars}$ and ${\hor_\aars} = {\hor_\bars}$,
  and every peak $c \veri_\beta a \hor_\alpha b$ in $\bars$ can be joined 
  by reductions of the form shown in Figure~\ref{fig:decreasing:diagram:commutation}.
  
  If all conditions are fulfilled, we call $\bars$ a \emph{decreasing labelling} of $\aars$.
\end{definition}

\begin{figure}[!ht]
  \centering
  \begin{tikzpicture}[default,thick,baseline=0ex,every node/.style={rectangle},inner sep=1mm]
    \node (a) {$a$};
    \node (b) [right of=a,node distance=35mm] {$b$};
    \node (c) [below of=a,node distance=35mm] {$c$};
    \node (d) [right of=c,node distance=35mm] {$d$};
    \node (b') [smallCircle] at ($(b)!.33!(d)$) {};
    \node (b'') [smallCircle] at ($(b)!.66!(d)$) {};
    \node (c') [smallCircle] at ($(c)!.33!(d)$) {};
    \node (c'') [smallCircle] at ($(c)!.66!(d)$) {};
    \draw [->,hor] (a) -- node [above] {$\alpha$} (b);
    \draw [->] (a) -- node [left] {$\beta$} (c);
    \begin{scope}[exi]
    \draw [->>] (b) -- node [right] {$< \alpha$} (b');
    \draw [->] (b') -- node [right] {$\beta$} node [left] {$\equiv$} (b'');
    \draw [->>] (b'') -- node [right,align=center] {$<\alpha$ $\cup$ $<\beta$} (d);
    \draw [->>,hor] (c) -- node [below,yshift=-.4mm] {$< \beta$} (c');
    \draw [->,hor] (c') -- node [below,yshift=-.4mm,align=center] {$\alpha$} node [above] {$\equiv$} (c'');
    \draw [->>,hor] (c'') -- node [below,yshift=-.4mm,align=center] {$<\alpha$\\[-.45ex]\ $\cup$ $<\beta$} (d);
    \end{scope}
  \end{tikzpicture}\vspace{-1.5ex}
  \caption{Decreasing elementary diagram for proving commutation.}
  \label{fig:decreasing:diagram:commutation}
\end{figure}

\begin{theorem}[Decreasing Diagrams for Commutation -- Van Oostrom~\cite{oost:1994b}]\label{thm:decreasing:diagrams:commutation}
  If an ARS $\aars = (A,{\ver},{\hor})$ is decreasing commuting, 
  then $\ver$ commutes with $\hor$. \qed
\end{theorem}

Analogous to the classes $\DCR_\alpha$ for confluence,
we introduce classes $\DC_\alpha$ for commutation.
\begin{definition}\label{def:dc:alpha}
  For ordinals $\alpha$,
  let $\DC_\alpha$ denote the class of ARSs $\aars = (A,{\ver},{\hor})$
  that are decreasing commuting (Definition~\ref{def:decreasing:diagrams:commutation}) with label set $\{\,\beta \mid \beta < \alpha\,\}$
  ordered by the usual order~$<$ on ordinals.
  We say that $\aars$
  has the property $\DC_\alpha$, denoted $\DC_\alpha(\aars)$, if $\aars \in \DC_\alpha$.
\end{definition}
In Definition~\ref{def:dc:alpha} it suffices to consider total orders
since every partial well-founded order can be transformed into a total well-founded order.
This transformation~\cite{endr:klop:2013} preserves the decreasing elementary diagrams and does not need the Axiom of Choice.

In order to show that the hierarchy for commutation does not collapse,
we inductively construct, for every $n \in \nat$, an ARS $\aars_n$ that is $\DC_{5n+1}$, but not $\DC_{n}$.

\begin{definition}
  For every $n \in \nat$ we define
  a tuple $\Phi_n = (\aars_n, a_1, a, c, b, b_1)$
  consisting of
  an ARS $\aars_n = (A_n,{\ver_n},{\hor_n})$ 
  and distinguished elements $a_1, a, c, b, b_1 \in A_n$ by induction on $n$:
  \begin{enumerate}
    \item 
      Let $\Phi_0 = (\aars_0, a_1, c, c, c, b_1)$ where $\aars_0$ is the ARS displayed in Figure~\ref{fig:base_case}.
    \item 
      Let $\Phi_n = (\aars_n, a, a', c, b', b)$.
      We obtain $\aars_{n+1}$ as an extension of $\aars_n$ as shown in Figure~\ref{fig:commutation}.
      The inner dark part with the darker background is $\aars_n$.
      The extension consists of the addition of fresh elements $a_1,\ldots,a_7$ and $b_1,\ldots,b_7$ and rewrite steps as shown in the figure.
      We define $\Phi_{n+1} = (\aars_{n+1}, a_1, a, c, b, b_1)$.
  \end{enumerate}
\end{definition}

\newcommand{\xsqueeze}{-3mm}
\newcommand{\ysqueeze}{2mm}
\begin{figure}[!ht]
  \centering
  \begin{minipage}[c]{.33\textwidth}
    \centering
    \begin{tikzpicture}[default,node distance=15mm,n/.style={smallCircle}]
      \node (c1) [outer sep=1mm] {$a_1$};
      \node (a1) [below left of=c1,yshift=\ysqueeze] {$a_2$};
      \node (a2) [below right of=c1,yshift=\ysqueeze] {$a_3$};
      \node (d) [below right of=a1,yshift=\ysqueeze] {$c$};
      \node (b1) [below left of=d,yshift=\ysqueeze] {$b_2$};
      \node (b2) [below right of=d,yshift=\ysqueeze] {$b_3$};
      \node (c2) [outer sep=1mm] [below right of=b1,yshift=\ysqueeze] {$b_1$};
      
      \begin{scope}[->]
        \draw (c1) to (a1);
        \draw [hor] (c1) to (a2);
        \draw (a1) to[bend left=20] (d);
        \draw [hor] (a1) to[bend right=20] (d);
        \draw (a2) to[bend right=20] (d);
        \draw [hor] (a2) to[bend left=20] (d);
        \draw [hor] (c2) to (b1);
        \draw (c2) to (b2);
        \draw (b1) to[bend right=20] (d);
        \draw [hor] (b1) to[bend left=20] (d);
        \draw [hor] (b2) to[bend right=20] (d);
        \draw (b2) to[bend left=20] (d);
      \end{scope}
      
      \begin{pgfonlayer}{background}
        \draw [rounded corners=2mm,fill=cblue!5,draw=cblue] ($(c1) + (-14mm,4mm)$) rectangle ($(c2) + (14mm,-4mm)$);
        \node at (c1) [minimum size=5mm,draw=cred,fill=cred!10] {};
        \node at (c2) [minimum size=5mm,draw=cred,fill=cred!10] {};
      \end{pgfonlayer}
    \end{tikzpicture}
    \caption{\\Base case: one label suffices.}
    \label{fig:base_case}
    \vspace{4ex}
  \end{minipage}%
  \begin{minipage}[c]{.66\textwidth}
    \centering
    \begin{tikzpicture}[default,node distance=15mm,n/.style={smallCircle}]
      \node [outer sep=2mm] (1) {$a_1$};
      \node (2) [above right of=1,xshift=\xsqueeze] {$a_2$};
      \node (3) [below right of=1,xshift=\xsqueeze] {$a_3$};
      \node (4) [below right of=2,xshift=\xsqueeze] {$a_4$};
      \node (5) [above right of=4,xshift=\xsqueeze] {$a_5$};
      \node (6) [below right of=4,xshift=\xsqueeze] {$a_6$};
      \node (7) [below right of=5,xshift=\xsqueeze] {$a_7$};
      
      \node (c1') [outer sep=2mm,above right of=7,xshift=6mm+\xsqueeze] {$a$};
      \node (c2') [outer sep=2mm,below right of=7,xshift=6mm+\xsqueeze] {$b$};
      \node (c1) at (c1') [inner sep=4mm] {};
      \node (c2) at (c2') [inner sep=4mm] {};
      \node (d) at ($(c1)!.5!(c2)$) {$c$};
  
      \node (b7) [below right of=c1,xshift=6mm+\xsqueeze] {$b_7$};
      \node (b5) [above right of=b7,xshift=\xsqueeze] {$b_5$};
      \node (b6) [below right of=b7,xshift=\xsqueeze] {$b_6$};
      \node (b4) [below right of=b5,xshift=\xsqueeze] {$b_4$};
      \node (b2) [above right of=b4,xshift=\xsqueeze] {$b_2$};
      \node (b3) [below right of=b4,xshift=\xsqueeze] {$b_3$};
      \node [outer sep=2mm] (b1) [below right of=b2,xshift=\xsqueeze] {$b_1$};
  
      \begin{scope}[->]
        \draw [hor] (1) -- (2);
        \draw [hor] (2) -- (4);
        \draw [hor] (4) -- (5);
        \draw [hor] (5) -- (7);
        \draw [hor] (7) -- (c1.-160);
        
        \draw (1) -- (3);
        \draw (3) -- (4);
        \draw (4) -- (6);
        \draw (6) -- (7);
        \draw (7) -- (c2.160);
  
        \draw [hor] (3) -- (6);
        \draw [hor] (6) -- (c2);
  
        \draw (2) -- (5);
        \draw (5) -- (c1);
      \end{scope}
      
      \begin{scope}[->]
        \draw (b1) -- (b2);
        \draw (b2) -- (b4);
        \draw (b4) -- (b5);
        \draw (b5) -- (b7);
        \draw (b7) -- (c1.-20);
        
        \draw [hor] (b1) -- (b3);
        \draw [hor] (b3) -- (b4);
        \draw [hor] (b4) -- (b6);
        \draw [hor] (b6) -- (b7);
        \draw [hor] (b7) -- (c2.20);
  
        \draw (b3) -- (b6);
        \draw (b6) -- (c2);
  
        \draw [hor] (b2) -- (b5);
        \draw [hor] (b5) -- (c1);
      \end{scope}
      
      \begin{scope}[->>]
        \draw [hor] (c1') to[bend left=20] (d);
        \draw [hor] (c2') to[bend left=20] (d);
        \draw (c1') to[bend left=-20] (d);
        \draw (c2') to[bend left=-20] (d);
      \end{scope}
  
      \begin{pgfonlayer}{background}
        \draw [rounded corners=2mm,fill=cblue!5,draw=cblue] ($(1) + (-4mm,16mm)$) rectangle ($(b1) + (4mm,-16mm)$);
        \draw [rounded corners=2mm,fill=cblue!20,draw=cblue] ($(c1) + (-4mm,4mm)$) rectangle ($(c2) + (4mm,-4mm)$);
        \node at (c1') [minimum size=5mm,draw=cred,fill=cred!10] {};
        \node at (c2') [minimum size=5mm,draw=cred,fill=cred!10] {};
        \node at (1) [minimum size=5mm,draw=cred,fill=cred!10] {};
        \node at (b1) [minimum size=5mm,draw=cred,fill=cred!10] {};
        
        \draw [cred,ultra thick] ($(1) + (2mm,0) + (40:3mm)$) arc (40:-40:3mm);
        \draw [cred,ultra thick] ($(4) + (2mm,0) + (40:3mm)$) arc (40:-40:3mm);
        \draw [cred,ultra thick] ($(7) + (2mm,0) + (40:3mm)$) arc (40:-40:3mm);
        \draw [cred,ultra thick] ($(b1) + (-2mm,0) + (40:-3mm)$) arc (40:-40:-3mm);
        \draw [cred,ultra thick] ($(b4) + (-2mm,0) + (40:-3mm)$) arc (40:-40:-3mm);
        \draw [cred,ultra thick] ($(b7) + (-2mm,0) + (40:-3mm)$) arc (40:-40:-3mm);
      \end{pgfonlayer}
    \end{tikzpicture}
    \caption{From $n$ to $n+1$ labels for commutation.
      Rough proof sketch: Assume that at least one of the reductions 
      $a \ver^* c$, $b \hor^* c$,
      $a \hor^* c$ or $b \ver^* c$
      contains two steps labelled with $n$.
      Then each of the peaks at $a_1$, $a_4$ and $a_7$,
      or each of the peaks at $b_1$, $b_4$ and $b_7$
      must contain a step labelled with $n+1$.
      As a consequence, one of the reductions
      $a_1 \ver^* c$, $b_1 \hor^* c$,
      $a_1 \hor^* c$ or $b_1 \ver^* c$
      contains two steps labelled with $n+1$.
      }
    \label{fig:commutation}
  \end{minipage}
\end{figure}

\noindent
We start with a few important properties of the construction.
\begin{lemma}\label{lem:commutation:properties}
  For every $n \in \nat$ and $\Phi_n = (\aars_n, a_1, a, c, b, b_1)$ with $\aars_n = (A_n,{\ver},{\hor})$ 
  we have the following properties:
  \begin{enumerate}[label=(\roman*)]
    \item \label{lem:commutation:properties:deterministic} The relations $\ver$ and $\hor$ are deterministic.
    \item \label{lem:commutation:properties:c} For every element $x \in A_n$ we have $x \ver^* c$ and $x \hor^* c$.
    \item \label{lem:commutation:properties:join:a} 
      For $x \in A_n$, we have $a_1 \hor^* x \veri^* b_1$ if and only if $a \hor^* x$ and $a \ver^* x$.
    \item \label{lem:commutation:properties:join:b} 
      For $x \in A_n$, we have $a_1 \ver^* x \hori^* b_1$ if and only if $b \hor^* x$ and $b \ver^* x$.
  \end{enumerate}
\end{lemma}

\begin{proof}
  We use induction on $n \in \nat$.
  For the base case $n = 0$, we have $\Phi_0 = (\aars_0, a_1, c, c, c, b_1)$ where $\aars_0$ is given in Figure~\ref{fig:base_case}.
  The properties follow from an inspection of the figure.
  
  For the induction step, let $n \in \nat$ and assume that $\Phi_n = (\aars_n, a, a', c, b', b)$ satisfies the properties.
  By construction, $\aars_{n+1}$ is an extension of $\aars_n$ as shown in Figure~\ref{fig:commutation}, and we have $\Phi_{n+1} = (\aars_{n+1}, a_1,a,c,b,b_1)$.
  The fresh elements introduced by the extension are $X = \{\, a_1,\ldots,a_7, b_1,\ldots,b_7 \,\}$.
  We check the validity of each property for $\aars_{n+1}$:
  \begin{enumerate}[label=(\roman*)]
    \item 
      There are no fresh steps with sources in $\aars_n$.
      Every element $x \in X$ admits precisely one outgoing step $\ver$ and one outgoing step $\hor$.
      So both rewrite relations remain deterministic, establishing property~\ref{lem:commutation:properties:deterministic}.  
    \item 
      For every element $x \in X$ we have $x \ver^* a$ or $x \ver^* b$, and $x \hor^* a$ or $x \hor^* b$.
      Together with the induction hypothesis~\ref{lem:commutation:properties:c} for $n$, this yields property~\ref{lem:commutation:properties:c} for $n + 1$.
    \item 
      From Figure~\ref{fig:commutation} it follows immediately that 
      any reduction $a_1 \hor^* x \veri^* b_1$ must be of the form $a_1 \hor^* a \hor^* x \veri^* a \veri^* b_1$.
      The reductions from both sides are deterministic and the first joining element is $a$.
    \item 
      Analogous to property~\ref{lem:commutation:properties:join:a}. \qedhere 
  \end{enumerate}
\renewcommand{\qedsymbol}{}
\end{proof}

\noindent
From Lemma~\ref{lem:commutation:properties}~\ref{lem:commutation:properties:c} 
it follows that $\ver$ and $\hor$ commute in $\aars_n$.
However, commutation is not sufficient to conclude that $\aars_n$ is decreasing commuting.
Decreasing diagrams are not complete for proving commutation as shown in~\cite{endr:klop:2013}.

We prove that $\aars_n$ is decreasing commuting by constructing a labelling with $5n$ labels.
This bound is by no means optimal, but easy to verify and sufficient for our purpose. 
\begin{lemma}\label{lem:commutation:positive}
  For every $n \in \nat$, $\aars_n$ is $\DC_{5n+1}$.
\end{lemma}

\begin{proof}
  We use induction on $n \in \nat$. 
  For the base case $n = 0$, consider $\aars_0$ shown in Figure~\ref{fig:base_case}.
  For this system a single label suffices since the joining reductions in the elementary diagrams have length at most $1$.

  For the induction step, assume that $\aars_n$ has the property $\DC_{5n+1}$.
  So $\aars_n$ is decreasing commuting with labels $\{\, 0,\ldots,\ell \,\}$ where $\ell = 5n$.
  By construction, $\aars_{n+1}$ is an extension of $\aars_n$ as shown in Figure~\ref{fig:commutation}.
  We extend the labelling of $\aars_n$ with labels $\{\, 0,\ldots,\ell \,\}$ to a labelling of  $\aars_{n+1}$ with labels $\{\, 0,\ldots,\ell+5 \,\}$ as follows:
  \begin{center}
  \begin{tikzpicture}[default,node distance=17mm,n/.style={smallCircle},l/.style={scale=0.8,rectangle,inner sep=1.3mm}]
    \node (1) {$a_1$};
    \node (2) [above right of=1] {$a_2$};
    \node (3) [below right of=1] {$a_3$};
    \node (4) [below right of=2] {$a_4$};
    \node (5) [above right of=4] {$a_5$};
    \node (6) [below right of=4] {$a_6$};
    \node (7) [below right of=5] {$a_7$};
    
    \node (c1') [outer sep=1mm,above right of=7,xshift=6mm] {$a$};
    \node (c2') [outer sep=1mm,below right of=7,xshift=6mm] {$b$};
    \node (c1) at (c1') [inner sep=4mm] {};
    \node (c2) at (c2') [inner sep=4mm] {};
    \node (d) at ($(c1)!.5!(c2)$) {$c$};

    \node (b7) [below right of=c1,xshift=6mm] {$b_7$};
    \node (b5) [above right of=b7] {$b_5$};
    \node (b6) [below right of=b7] {$b_6$};
    \node (b4) [below right of=b5] {$b_4$};
    \node (b2) [above right of=b4] {$b_2$};
    \node (b3) [below right of=b4] {$b_3$};
    \node (b1) [below right of=b2] {$b_1$};

    \begin{scope}[->]
      \draw [hor] (1) -- node [l,sloped,below] {$\ell+5$} (2);
      \draw [hor] (2) -- node [l,sloped,below] {$\ell+4$} (4);
      \draw [hor] (4) -- node [l,sloped,below] {$\ell+3$} (5);
      \draw [hor] (5) -- node [l,sloped,below] {$\ell+2$} (7);
      \draw [hor] (7) -- node [l,sloped,below] {$\ell+1$} (c1.-160);
      
      \draw (1) -- node [l,sloped,above] {$\ell+5$} (3);
      \draw (3) -- node [l,sloped,above] {$\ell+4$} (4);
      \draw (4) -- node [l,sloped,above] {$\ell+3$} (6);
      \draw (6) -- node [l,sloped,above] {$\ell+2$} (7);
      \draw (7) -- node [l,sloped,above] {$\ell+1$} (c2.160);

      \draw (2) -- node [l,sloped,above] {$\ell+4$} (5);
      \draw (5) -- node [l,sloped,above] {$\ell+2$} (c1);

      \draw [hor] (3) -- node [l,sloped,below] {$\ell+4$} (6);
      \draw [hor] (6) -- node [l,sloped,below] {$\ell+2$} (c2);
    \end{scope}
    
    \begin{scope}[->]
      \draw (b1) -- node [l,sloped,below] {$\ell+5$} (b2);
      \draw (b2) -- node [l,sloped,below] {$\ell+4$} (b4);
      \draw (b4) -- node [l,sloped,below] {$\ell+3$} (b5);
      \draw (b5) -- node [l,sloped,below] {$\ell+2$} (b7);
      \draw (b7) -- node [l,sloped,below] {$\ell+1$} (c1.-20);
      
      \draw [hor] (b1) -- node [l,sloped,above] {$\ell+5$} (b3);
      \draw [hor] (b3) -- node [l,sloped,above] {$\ell+4$} (b4);
      \draw [hor] (b4) -- node [l,sloped,above] {$\ell+3$} (b6);
      \draw [hor] (b6) -- node [l,sloped,above] {$\ell+2$} (b7);
      \draw [hor] (b7) -- node [l,sloped,above] {$\ell+1$} (c2.20);

      \draw (b3) -- node [l,sloped,below] {$\ell+4$} (b6);
      \draw (b6) -- node [l,sloped,below] {$\ell+2$} (c2);

      \draw [hor] (b2) -- node [l,sloped,above] {$\ell+4$} (b5);
      \draw [hor] (b5) -- node [l,sloped,above] {$\ell+2$} (c1);
    \end{scope}
    
    \begin{scope}[->>]
      \draw [hor] (c1') to[bend left=20] (d);
      \draw [hor] (c2') to[bend left=20] (d);
      \draw (c1') to[bend left=-20] (d);
      \draw (c2') to[bend left=-20] (d);
    \end{scope}

    \begin{pgfonlayer}{background}
      \draw [rounded corners=2mm,fill=cblue!20,draw=cblue] ($(c1) + (-4mm,4mm)$) rectangle ($(c2) + (4mm,-4mm)$);
    \end{pgfonlayer}
  \end{tikzpicture}
  \end{center}
  Here $\aars_n$ is the darker inner part.
  From the picture it is easy to verify that every peak $\veri \cdot \hor$ in the extension can be joined 
  by reductions that only contain labels strictly smaller than labels of the peak.
  As a consequence, $\aars_{n+1}$ is $\DC_{5(n+1)+1}$.
\end{proof}

Next, we show that $\aars_n$ does not admit a decreasing labelling with $n$ labels.

\begin{lemma}\label{lem:commutation:negative}
  For every $n \in \nat$, $\aars_n$ is not $\DC_{n}$.
\end{lemma}

\begin{proof}
  We prove the following stronger claim:
  for every $n\in\nat$ and $\Phi_{n} = (\aars_n, a_1,a,c,b,b_1)$,
  and every decreasing labelling of $\aars_n$ with labels from $\nat$
  it holds that 
  at least one of the four paths $a_1 \ver^* b$, $a_1 \hor^* a$, $b_1 \ver^* a$ or $b_1 \hor^* b$
  contains two labels $\ge n$.
  Note that these paths exist by Lemma~\ref{lem:commutation:properties}.
  We prove this claim by induction on $n \in \nat$.

  For the base case $n = 0$, we have $\Phi_0 = (\aars_0, a_1, c, c, c, b_1)$ where $\aars_0$ is given in Figure~\ref{fig:base_case}.
  It suffices to consider one of the four paths. 
  For instance, the rewrite sequence $a_1 \ver^* c$ has length $2$ and both steps must have a label $\ge 0$.
  
  For the induction step, assume that the claim holds for $n$ and $\Phi_{n} = (\aars_{n}, a,a',c,b',b)$.
  Accordingly, the induction hypothesis is that, 
  for every decreasing labelling of $\aars_n$ with labels from $\nat$,
  one of the four paths
  $a \ver^* b'$, $a \hor^* a'$, $b \ver^* a'$ or $b \hor^* b'$
  contains two labels $\ge n$.
  We prove the claim for $n + 1$.
  Let $\Phi_{n+1} = (\aars_{n+1}, a_1,a,c,b,b_1)$ where $\aars_{n+1}$ is an extension of $\aars_n$
  as shown in Figure~\ref{fig:commutation}.
  Let $\bars$ be a decreasing labelling of the steps in $\aars_{n+1}$ with labels from $\nat$.
  We show that at least one of the paths
  $a_1 \ver^* b$, $a_1 \hor^* a$, $b_1 \ver^* a$ or $b_1 \hor^* b$
  contains two labels $\ge n+1$.

  By construction, the systems $\aars_{n+1}$ and $\aars_n$ contain the same steps with sources in~$\aars_n$.
  Thus the restriction of the labelling $\bars$ to $\aars_n$ is a decreasing labelling for $\aars_n$.
  By the induction hypothesis, 
  at least one of the paths (i) $a \ver^* b'$, (ii) $a \hor^* a'$, (iii) $b \ver^* a'$ or (iv) $b \hor^* b'$
  contains two labels $\ge n$.
  Without loss of generality, by symmetry, assume that the path (i) or (iv) contain two labels $\ge n$.

  Consider the peak $a_3 \veri a_1 \hor a_2$.
  As visible in Figure~\ref{fig:commutation},
  every elementary diagram for this peak must have joining reductions of the form
  $a_3 \hor^* b \hor^* x \veri^* a \veri^* a_2$ for some $x \in \aars_n$.
  From Lemma~\ref{lem:commutation:properties}~\ref{lem:commutation:properties:join:b}
  we conclude that the joining reductions must be of the form
  \begin{align*}
    a_3 \hor^* b \hor^* b' \hor^* x \veri^* b' \veri^* a \veri^* a_2
  \end{align*}
  The path (i) $a \ver^* b'$ or (iv) $b \hor^* b'$ contains two labels $\ge n$.
  Thus, for the elementary diagram to be decreasing,
  one of the steps in the peak $a_3 \veri a_1 \hor a_2$ must have label $\ge n+1$.

  The same argument can be applied to the peaks $a_6 \veri a_4 \hor a_5$ and $b \veri a_7 \hor a$.
  As a consequence, each of the peaks
  $a_3 \veri a_1 \hor a_2$, $a_6 \veri a_4 \hor a_5$ and $b \veri a_7 \hor a$
  contains one step with a label $\ge n + 1$.
  Hence at least one of the paths 
  \begin{enumerate}
    \item $a_1 \ver a_3 \ver a_4 \ver a_6 \ver a_7 \ver b$, or
    \item $a_1 \hor a_2 \hor a_4 \hor a_5 \hor a_7 \hor a$
  \end{enumerate}
  contains two steps with labels $\ge n+1$.
%
%
  This proves the claim and concludes the proof.
\end{proof}

We have seen that, for every $n\in\nat$, $\aars_n$ that is $\DC_{5n+1}$, but not $\DC_{n}$ (Lemmas~\ref{lem:commutation:positive} \&~\ref{lem:commutation:negative}).
From this we can conclude that an infinite number of the inclusions
$\DC_0 \subseteq \DC_1 \subseteq \DC_2 \subseteq \cdots$ 
are strict.
The following proposition allows us to infer that all of them are strict.

Roughly speaking, the following proposition states that if a level $\alpha + 1$ of the hierarchy does not collapse,
then also the level $\alpha$ does not collapse.
We state the proposition for the commutation hierarchy, but it also holds for the confluence hierarchy. 
\begin{proposition}\label{prop:downward}
  If $\DC_{\alpha} \subsetneq \DC_{\alpha+1}$ for an ordinal $\alpha$, then $\DC_{\beta} \subsetneq \DC_{\alpha}$ for every $\beta < \alpha$.
  This also holds when the classes are restricted to countable systems.
\end{proposition}

\begin{proof}
  Let $\aars = (A,{\ver},{\hor})$ be in $\DC_{\alpha+1} \setminus \DC_{\alpha}$.
  Then there exists a decreasing labelling $\bars$ of $\aars$ with labels $\{\, \beta \mid \beta \le \alpha \,\}$.
  As $\aars$ is not $\DC_{\alpha}$ some steps must have the maximum label~$\alpha$.
  Note that
  \begin{itemize}
    \item [$\star$]
      If the joining reductions in a decreasing elementary diagram contain a step with label $\alpha$,
      then the corresponding peak must also contain a step with label~$\alpha$. 
  \end{itemize}
  Let $\bars'$ be obtained from $\bars$ by dropping all steps with label $\alpha$,
  and let $\aars'$ be obtained from $\bars'$ by dropping the labels.
  By $(\star)$, $\bars'$ is a decreasing labelling of $\aars'$, and hence $\aars'$ is $\DC_{\alpha}$.

  For a contradiction, assume that $\DC_{\beta} = \DC_{\alpha}$ for some $\beta < \alpha$.
  Then $\aars'$ is $\DC_{\beta}$.
  Let $\bars''$ be obtained from $\bars'$ by adding all steps that we had previously removed from $\bars$,
  but we now relabel the steps from $\alpha$ to $\beta$.
  It is straightforward to check that $\bars''$ is a decreasing labelling of $\aars$.
  Hence, $\aars$ is in $\DC_{\beta+1} \subseteq \DC_{\alpha}$. This is a contradiction.
\end{proof}

\noindent
\begin{example}
  Assume that $\alpha$ is a limit ordinal and $\DC_{\alpha+3} \subsetneq \DC_{\alpha+4}$.
  By Proposition~\ref{prop:downward} we conclude $\DC_{\alpha+2} \subsetneq \DC_{\alpha+3}$.
  By repeated application of Proposition~\ref{prop:downward} we conclude 
  \begin{align*}
    \DC_{\beta} \subsetneq \DC_{\alpha} \subsetneq \DC_{\alpha+1} \subsetneq \DC_{\alpha+2} \subsetneq \DC_{\alpha+3} \subsetneq \DC_{\alpha+4}
  \end{align*}
  for every $\beta < \alpha$.
  However, the proposition does not help to conclude that $\DC_{\beta} \subsetneq \DC_{\beta'}$ for every $\beta < \beta' \le \alpha$.
\end{example}

\begin{theorem}\label{thm:dc:alpha}
  We have
  \begin{enumerate}[label=(\roman*)]
    \item $\DC_n \subsetneq \DC_{n+1}$ for every $n \in \nat$, and
    \item $\bigcup_{n \in \nat} \DC_n \subsetneq \DC_\omega$.
  \end{enumerate}
  These inclusions are strict also when the classes are restricted to countable systems.
\end{theorem}

\begin{proof}
  By Lemmas~\ref{lem:commutation:positive} and~\ref{lem:commutation:negative} we know that
  $\DC_n \subsetneq \DC_{n+1}$ for infinitely many $n \in \nat$.
  Then repeated application of Proposition~\ref{prop:downward} yields $\DC_n \subsetneq \DC_{n+1}$ for every $n \in \nat$.
  
  Let $\aars$ be the infinite disjoint union $\aars_0 \uplus \aars_1 \uplus \aars_2 \uplus \cdots$.
  As a consequence of Lemmas~\ref{lem:commutation:positive} and~\ref{lem:commutation:negative}
  the ARS $\aars$ is $\DC_\omega$ but not $\DC_n$ for any $n\in \nat$.
\end{proof}


%% file: conclusion.tex
\newcommand{\shiftdraw}[5][]{
  \draw [draw=none] (#2) -- coordinate[pos=0] (start) coordinate[pos=1] (end) (#3); 
  \coordinate (start) at ([shift={#4}]start);
  \coordinate (end) at ([shift={#4}]end);
  \draw [#1] (start) -- #5 (end);
}
\newcommand{\hcross}[2][]{
  \draw [draw=none] (start) -- coordinate[pos=#2] (x) (end); 
  \draw [#1] ([shift={(-1mm,-1mm)}]x) -- ([shift={(1mm,1mm)}]x);
}

In this paper we were concerned with the general question whether for abstract rewrite systems 
we could establish a hierarchy of complexity concerning the confluence property 
of abstract rewrite systems.

This led us first in this paper, in Section~\ref{sec:first:order}, to an investigation of the first-order definability of these various reduction properties: not only confluence and termination, but also several more, such as strong confluence and inductivity -- in total over a dozen of properties. 
The rationale of this scrutiny of first-order definability is that definability by a set of first-order formulas 
would possibly enable us to detect a hierarchy of complexity by imposing syntactic restrictions on such a defining set of formulas.
This section is considerably extended as compared to the conference proceedings version of this paper~\cite{endr:klop:over:2018}, of which the current paper is an extension. 
This section on first-order definability, with its introduction of finite model theory methods for abstract rewriting theory, 
can be considered as the main part of the current extended paper.

We pose the following open problem, which was suggested to us by one of the referees of the current paper:

\begin{open}
 The properties \UN, \UNrew and \AC{} turn out to be gfops. What is the intuition behind this fact? Is the fact that they do not have reachability in the consequent of their implication relevant? Can one give even  a classification of gfop properties that are formulated in the signature as employed for the considered properties?
\end{open}

Next, in Sections~\ref{sec:two} and~\ref{sec:commutation}, we continue with the study (as in the original paper as mentioned), of decreasing diagrams, in particular how the strength of decreasing diagrams is influenced by the size of the label set. 
We find that all abstract rewrite systems with the cofinality property
(in particular, all confluent, countable systems) can be proven confluent 
using the decreasing diagrams technique
with the almost trivial label set $I = \{\,0,1\,\}$.\footnote{
  Our results have found applications in~\cite[Lemma 1 \& Remark 3]{hiro:nage:oost:oyam:2019}.
}
So for confluence of \emph{countable} ARSs, we have the following implications:
\begin{center}
  \medskip
  \begin{tikzpicture}[default,thick,baseline=0ex,every node/.style={rectangle,scale=1,outer sep=0mm,inner sep=1mm},inner sep=1mm]
    \node (cp) at (4cm,4cm) {$\CP$};
    \node (cr) at (4cm,0cm) {$\CR$};
    \node (dcr) at (0cm,0cm) {$\DCR$};
    \node (dcr2) [rounded corners=1mm,fill=cblue!20,outer sep=1mm] at (0cm,4cm) {$\DCR_2$};

    \node (dcr1) at ($(dcr2)!-.23!(dcr)$) {$\DCR_1$};
    \node (dcr3) at ($(dcr2)!.23!(dcr)$) {$\DCR_3$};
    \node (dcr4) at ($(dcr2)!.42!(dcr)$) [rotate=0,inner sep=-1mm] {\raisebox{1.5mm}{$\vdots$}};
    \node (dcr5) at ($(dcr2)!.61!(dcr)$) {$\DCR_\omega$};
    \node (dcr6) at ($(dcr2)!.80!(dcr)$) [rotate=0,inner sep=-1mm] {\raisebox{1.5mm}{$\vdots$}};
    
    \draw [->] (cp) to (dcr2);
    \shiftdraw[->]{dcr1}{dcr2}{(1.5mm,0mm)}{}
    \shiftdraw[->,black!50]{dcr2}{dcr1}{(-1.5mm,0mm)}{} \hcross[black!50]{0.25}
    \draw [->] (dcr2) to (dcr3);
    \draw [->] (dcr3) to (dcr4);
    \draw [->] (dcr4) to (dcr5);
    \draw [->] (dcr5) to (dcr6);
    \draw [->] (dcr6) to (dcr);
    \draw [->] (dcr) to (cr);
    \draw [->] (cr) to (cp);
    
    \node at (2cm,2cm) {$\aars$ countable};
  \end{tikzpicture}
  \medskip
\end{center}
This is in sharp contrast to the situation for commutation for which we prove
\begin{align*}
  \DC_1 \subsetneq \DC_2 \subsetneq \DC_3 \subsetneq \cdots \subsetneq \DC_\omega
\end{align*}
even for countable systems.
So for commutation, for every $n \le \omega$, there exists a system that requires $n$ labels.
The structure of this hierarchy above level $\omega$ remains open.

\begin{open}
  What inclusions $\DC_\alpha \subseteq \DC_\beta$ are strict for $\omega \le \alpha < \beta$?
\end{open}

Decreasing diagrams are complete for confluence of countable systems.
However, 
it is a long-standing open problem whether the method of decreasing diagrams 
is also complete for proving confluence of uncountable systems~\cite{oost:1994b}.
Our observations may provide new ways for approaching this problem.
In particular, it may be helpful to investigate the following:

\begin{open}
  Is there a confluent, uncountable system that is $\CR$ but not $\DCR_2$?
\end{open}

\begin{open}
  Is there a confluent, uncountable system that needs more than $2$ labels to establish confluence using decreasing diagrams?
  In other words, is there an uncountable system that is $\DCR$ but not $\DCR_2$?
  Is there an uncountable system that is $\DCR_3$ but not $\DCR_2$?
\end{open}

So we have the following situation for uncountable systems\footnote{%
  Note that the implication $\DCR_1 \implies \CP$ fails.
  To see this, consider the ARS $(2^\real,\to)$ where the steps are of the form $X \to X \cup \{\,y\,\}$ for $X \subseteq \real$ and $y \in \real$. 
}:
%
\begin{center}
  \medskip
  \begin{tikzpicture}[default,thick,baseline=0ex,every node/.style={rectangle,scale=1,outer sep=0mm,inner sep=1mm},inner sep=1mm]
    \node (cp) at (4cm,4cm) {$\CP$};
    \node (cr) at (4cm,0cm) {$\CR$};
    \node (dcr) at (0cm,0cm) {$\DCR$};
    \node (dcr2) [rounded corners=1mm,fill=cblue!20,outer sep=1mm] at (0cm,4cm) {$\DCR_2$};

    \node (dcr1) at ($(dcr2)!-.25!(dcr)$) {$\DCR_1$};
    \node (dcr3) at ($(dcr2)!.25!(dcr)$) {$\DCR_3$};
    \node (dcr4) at ($(dcr2)!.5!(dcr)$) {$\DCR_\alpha$};
    \node (dcr5) at ($(dcr2)!.75!(dcr)$) {$\DCR_\beta$};
    
    \shiftdraw[->]{cp}{dcr2}{(0mm,1.5mm)}{} 
    \shiftdraw[->,black!50]{dcr2}{cp}{(0mm,-1.5mm)}{} \hcross[black!50]{0.5}

    \shiftdraw[->]{dcr1}{dcr2}{(1.5mm,0mm)}{}
    \shiftdraw[->,black!50]{dcr2}{dcr1}{(-1.5mm,0mm)}{} \hcross[black!50]{0.25}
    \shiftdraw[->]{dcr2}{dcr3}{(1.5mm,0mm)}{}
    \shiftdraw[->]{dcr3}{dcr4}{(1.5mm,0mm)}{}
    \shiftdraw[->]{dcr4}{dcr5}{(1.5mm,0mm)}{}
    \shiftdraw[->]{dcr5}{dcr}{(1.5mm,0mm)}{}
    \shiftdraw[<-,red]{dcr2}{dcr3}{(-1.5mm,0mm)}{node [left,scale=0.8,inner sep=2mm] {?}}
    \shiftdraw[<-,red]{dcr3}{dcr4}{(-1.5mm,0mm)}{node [left,scale=0.8,inner sep=2mm] {?}}
    \shiftdraw[<-,red]{dcr4}{dcr5}{(-1.5mm,0mm)}{node [left,scale=0.8,inner sep=2mm] {?}}
    \shiftdraw[<-,red]{dcr5}{dcr}{(-1.5mm,0mm)}{node [left,scale=0.8,inner sep=2mm] {?}}

    \shiftdraw[->]{dcr}{cr}{(0mm,-1.5mm)}{} 
    \shiftdraw[->,red]{cr}{dcr}{(0mm,1.5mm)}{node [above,scale=0.8,inner sep=2mm] {?}} 
    
    \shiftdraw[->]{cp}{cr}{(1.5mm,0mm)}{}
    \shiftdraw[->,black!50]{cr}{cp}{(-1.5mm,0mm)}{} \hcross[black!50]{0.5}

    \node at (2cm,2cm) {$\aars$ uncountable};
  \end{tikzpicture}
  \medskip
\end{center}
Here the question marks indicate open problems.


For a better understanding of this hierarchy, it would be interesting to investigate whether Proposition~\ref{prop:downward} can be generalised as follows.
\begin{open}
  Assume that $\DC_{\alpha} \subsetneq \DC_{\beta}$ for ordinals $\alpha < \beta$.
  Does this imply that none of the lower levels of the hierarchy collapse?
  That is, does it imply that $\DC_{\alpha'} \subsetneq \DC_{\beta'}$ for every $\alpha' < \beta' \le \alpha$?
\end{open}

Our findings indicate that the size of the label set in decreasing diagrams is not a suitable measure for the complexity of a confluence problem.
So the complexity arises rather from the distribution of the labels, 
and the proof that every peak has suitable joining reductions. 
The complexity of the label distribution can be measured in terms of the complexity of machine required for computing the labels. 
For this purpose, one can consider 
Turing machines, finite automata or finite state transducers.
The complexity of Turing machines can be measured in terms of 
time or space complexity, Kolmogorov Complexity~\cite{li:vita:2008} or degrees of unsolvability~\cite{shoe:1971}.
For finite state transducers the complexity can be classified by degrees of transducibility~\cite{endr:hend:klop:2011, endr:karh:klop:saar:2016, endr:klop:saar:whit:2015}.

Another interesting matter with respect to first-order definability,
is to consider the case of \emph{two} relations, blue and red,
and consider properties such as the \emph{jumping property}~\cite{ders:2012,oost:zant:2012} for such pairs of reduction relations.
 
\section*{Tiling for uncountable systems}

For us the most fundamental open problem is the following.  
As we have seen for countable systems, the question of confluence can always be reduced to local confluence.
This means that every confluence diagram can always be fully tiled by elementary local confluence diagrams.
For uncountable systems this question is wide open. 
It is conceivable that there exist complicated uncountable systems whose confluence is
due to quite other properties than local confluence.
Then confluence diagrams would not be `finitely tilable'.
Confluence then could `transcend' the procedure of locally adding tiles. 

\begin{open}
  Is there a confluence diagram in an uncountable ARS 
  that cannot be finitely tiled by elementary local confluence diagrams?
\end{open}
For commutation it has been shown in~\cite{endr:klop:2013}
that there exist commutation diagrams that cannot be finitely tiled by 
local commutation diagrams.

\subsection*{Acknowledgements}
We thank Vincent van Oostrom and Bertram Felgenhauer for many useful comments.
We are also thankful to Bertram for presenting an early version of this paper at the \emph{International Workshop on Confluence}
when none of the authors was able to attend.
Finally, we are thankful to the reviewers, of both the conference version and this extended version, for many useful suggestions.
Endrullis and Overbeek received funding from the Netherlands
Organization for Scientific Research (NWO) under the
Innovational Research Incentives Scheme Vidi (project. No. VI.Vidi.192.004)
and the COMMIT2DATA program (project No. 628.011.003, ECiDA), respectively.

